\documentclass{article}

\usepackage{custom-preamble}
\usepackage{custom-todo}
\usepackage{authblk}
\usepackage{booktabs}
\usepackage{tikz}

\usepackage{xcolor}
\definecolor{mypink}{rgb}{0.9, 0.0, 0.4}
\definecolor{mygreen}{rgb}{0.0,0.5,0.0}
\definecolor{mypurple}{rgb}{0.6,0.0,0.6}
\definecolor{myfuchsia}{rgb}{0.9,0.0,0.9}
\definecolor{mybrown}{rgb}{0.7,0.3,0.3}
\definecolor{mygray}{rgb}{0.6,0.6,0.6}
\definecolor{myltgray}{rgb}{0.85,0.85,0.85}
\definecolor{stnavy}{HTML}{0D1CC9}
\definecolor{teal}{rgb}{0.0, 0.5, 0.5}

\title{The nearest neighbor interchange proximity measure for the tree-basedness of rooted binary phylogenetic networks}
\author[1,*]{Martin Frohn}
\author[2]{Takatora Suzuki}

\affil[1]{Department of Advanced Computing Sciences, Maastricht University, Maastricht, 6229~EN, The Netherlands}
\affil[2]{Department of Applied Mathematics, Faculty of Fundamental Science and Engineering, Waseda University, Japan}
\affil[*]{Corresponding Author. Email: martin.frohn@maastrichtuniversity.nl}

{\bf}{}
\newcommand{\head}{\mathrm{head}}
\newcommand{\tail}{\mathrm{tail}}

\begin{document}
\maketitle

\noindent\makebox[\linewidth]{\rule{\textwidth}{1pt}} \\

\begin{abstract}
Tree-based networks are a broad graph class in phylogenetics which posit a model of evolution composed of tree-like processes and hybridization events. Hence, for any non-tree-based phylogenetic network, it is meaningful to study the magnitude of violation of the tree-based model assumptions. For rooted networks, we introduce the proximity measure $\delta_{\text{rNNI}}$, defined by the minimum number of rooted nearest neighbor interchanges (NNIs) required to transform a network into a tree-based network. For binary networks, we prove that $\delta_{\text{rNNI}}$ is pairwise distinct from all previously known proximity measures in the literature and that no approximation guarantees between $\delta_{\text{rNNI}}$ and other known proximity measures holds. In addition, we characterize a family of networks for which a linear time lower bound on $\delta_{\text{rNNI}}$ is tight. Moreover, we analyze upper bounds and give an exact solution algorithm to compute $\delta_{\text{rNNI}}$.
\end{abstract}

\begin{keywords}
    Tree-based phylogenetic network, proximity measure, nearest neighbor interchange
\end{keywords}

\noindent\makebox[\linewidth]{\rule{\textwidth}{1pt}} \\

\section{Introduction}
Directed weighted acyclic graphs which model evolutionary histories of a set of species are called \emph{phylogenetic networks} if the leaves are bijectively labeled with the set of species. In the study of evolutionary relationships species are called \emph{taxa} and it is oftentimes assumed that phylogenetic networks are \emph{rooted}, i.e., a unique vertex has no parents and is called the \emph{root} of the network. 
The evolutionary model constituted by a phylogenetic network that is a tree whose internal vertices have at least degree three is also called a \emph{phylogenetic tree}. The fundamental class of \emph{tree-based} phylogenetic networks $N$ generalizes the notion of phylogenetic trees by requiring $N$ to contain a tree $T$ on the same set of taxa and with the same root as $N$ (if it exists), called the \emph{support tree} of $N$, such that $T$ has the same vertex set as $N$~\cite{kong2022classes}. Notice that a support tree does not adhere to any degree requirements on vertices like phylogenetic trees. Tree-based networks appear to form a proper generalization of phylogenetic trees by allowing for biological phenomena like horizontal gene transfer and incomplete lineage sorting which are absent from any phylogenetic tree model~\cite{bapteste2013networks}. However, using them as a model of evolution can exclude biological processes which are inherently not tree-like~\cite{dagan2006tree,doolittle2007pattern}. Consequently, the question arises to which degree a phylogenetic network has a tree-based graph structure.

The proximity of a phylogenetic networks to being tree-based can be measured in many different ways. For example by counting the number of vertices deleted from or added to a network to make it tree-based~\cite{francis2018new}, by counting the number of deleted or added edges to make a network tree-based or by modifying it with topological rearrangement moves ~\cite{fischer2020tree}. Connections between these different proximity measures and algorithms to compute them are typically derived from characterizations of tree-based phylogenetic networks. Specifically, the tree-based property is related to more general, well-known notions in graph theory such as bipartite matchings~\cite{zhang2016tree} and vertex-disjoint path partitions~\cite{francis2018new}.

In this article, we study nearest neighbor interchanges (NNIs) as a proximity measure $\delta_{\text{rNNI}}$ of tree-basedness for rooted binary phylogenetic networks. To the best of our knowledge a NNI proximity measure for rooted phylogenetic network has not been studied before while the unrooted version has recevied some attention in the literature~\cite{fischer2020tree}. We show that all known tree-basedness proximity measures for rooted phylogenetic networks from the literature do guarantee only a linear approximation of $\delta_{\text{rNNI}}$. To prove this distinguishing feature of $\delta_{\text{rNNI}}$ we leverage the characterization of rooted binary phylogenetic networks $N$ as a partition of zig-zag-paths~\cite{hayamizu2021structure}. Furthermore, we introduce lower and upper bounds on $\delta_{\text{rNNI}}$ by extending Theorems~11 and~19 in~\cite{suzuki2024bridging} and we analyse the families of networks on which these bounds are tight and not tight.

In~\cref{sec:prelim} we introduce notation to formalize the measure $\delta_{\text{rNNI}}$ and provide background information on tree-basedness proximity measures. In~\cref{sec:rNTBPcomb} we show that $\delta_{\text{rNNI}}$ is pairwise distinct from the known tree-basedness proximity measures for binary rooted phylogenetic networks and we discuss an algorithm to compute $\delta_{\text{rNNI}}$. In~\cref{sec:future} we give a conclusion and future research directions.


\section{Preliminaries on phylogenetic networks and tree-basedness proximity measures}\label{sec:prelim}
In this section we consider various preliminaries which we will use throughout this article.
\subsection{Graph theoretical terminology}
Throughout this paper, we assume that the graphs are finite, directed, simple, and acyclic, unless otherwise stated. 
For a graph $G$ we denote $V(G)$ and $E(G)$ as the vertex set and edge set of $G$, respectively. 
For two graphs~$G$ and~$H$, $G$ is a \emph{subgraph} of $H$ if both $V(G)\subseteq V(H)$ and $E(G)\subseteq E(H)$ hold, in which case we write $G\subseteq H$. 
Two directed graphs $G$ and $H$ are said to be \emph{isomorphic} if there exists a bijection between $V(G)$ and $V(H)$ that yields a bijection between $E(G)$ and $E(H)$.  
In other words, $G$ and $H$ are isomorphic if there exists a labelling map $f: V(G)\to V(H)$ so that there is a one-to-one correspondence between $V(G)$ and $V(H)$ and for any $v_i$ and $v_j$ in $V(G)$, $(v_i, v_j)$ is in $E(G)$ if and only if  $(f(v_i),  f(v_j))$ is in $E(H)$, including the loops and parallel edges if they exist. 
A subgraph $G$ of $H$ is \emph{proper} if $G\neq H$.

Given a graph $G$ and a non-empty subset $S \subseteq E(G)$, the edge-set $S$ is said to \emph{induce the subgraph  $G[S]$ of $G$}, that is, the one whose edge-set is $S$ and whose vertex-set is the set of the ends of all edges in $S$. 
For a graph $G$ with $|E(G)| \geq 1$ and a partition  $\{E_1,\dots,E_{d}\}$ of $E(G)$, the collection  $\{G[E_1], \dots, G[E_{d}]\}$ is a  \emph{decomposition} of $G$, and $G$ is said to be \emph{decomposed into} $\{G[E_1], \dots, G[E_{d}]\}$. 

For an edge $e=(u,v)$ of a graph $G$, $u$ and $v$ are denoted by $\tail(e)$ and $\head(e)$, respectively. For a graph~$G$ and $v\in V(G)$, we denote $d_{\text{in}}(v)$ and $d_{\text{out}}(v)$ as the number of edges $e$ of $G$ with $\head(e) = v$ and $\tail(e) = v$, respectively. 
We call $G$ \emph{binary} and \emph{almost-binary} if $d_{\text{in}}(v)+d_{\text{out}}(v)\leq 3$ and $d_{\text{in}}(v)+d_{\text{out}}(v)\leq 4$, respectively, for all $v\in V(G)$. 
For an edge $(u, v)$ of a directed graph $G$,  \emph{subdividing edge $(u, v)$} refers to the operation of introducing a new vertex $p$ and replacing $(u, v)$ with the new consecutive edges $(u, p)$ and $(p, v)$. Conversely, for a vertex $p\in V(G)$ with $d_{\text{in}}(p)=d_{\text{out}}(p)=1$, we call $p$ an \emph{edge subdivision} of $G$. Similarly, subdividing an edge $\{u,v\}$ of an undirected graph $G$ refers to the replacement of $\{u,v\}$ by two edges $\{u,p\},\{p,v\}$ for a new vertex $p$, and a edge subdivision $p\in V(G)$ is defined by $\text{deg}(p)=2$. 

Two directed graphs $G$ and $H$ are said to be \emph{(graph) isomorphic} if there exists a bijection between $V(G)$ and $V(H)$ that yields a bijection between $E(G)$ and $E(H)$.  
In other words, $G$ and $H$ are isomorphic if there exists a labelling map $f: V(G)\to V(H)$ so that there is a one-to-one correspondence between $V(G)$ and $V(H)$ and for any $v_i$ and $v_j$ in $V(G)$, $(v_i, v_j)$ is in $E(G)$ if and only if  $(f(v_i),  f(v_j))$ is in $E(H)$, including the loops and parallel edges if they exist. 

\subsection{Phylogenetic networks}
We complete the definition of a phylogenetic network $N$ from the introduction by requiring that the underlying directed weighted acyclic graph of $N$ contains no edge subdivisions. Since edge lengths and the bijection between leaves and species are not a focal point in this article we refer to phylogenetic networks as their underlying graphs if not explicitly stated otherwise. If a phylogenetic network $N$ is rooted, then we denote the unique root of $N$ by $\rho(N)$. By definition $d_{\text{in}}(\rho(N))=0$ and there exists a directed path from $\rho(G)$ to any vertex $v$ with $d_{\text{in}}(v)=1$ and $d_{\text{out}}(v)=0$, called a \emph{leaf} of $N$. For a rooted phylogenetic network $N$, we call $v\in V(N)$ a \emph{reticulation vertex} if $d_{\text{in}}(v)\geq 2$ and $d_{\text{out}}(v)=1$, a \emph{internal tree vertex} if $d_{\text{in}}(v)=1$ and $d_{\text{out}}(v)\geq 2$ and a \emph{leaf} if $d_{\text{in}}(v)=1$ and $d_{\text{out}}(v)=0$. 
Notice that every unrooted binary phylogenetic network can be transformed into a rooted binary phylogenetic network and vice versa by deleting one edge incident to a leaf and orienting edges with respect to the binary constraint or attaching a new leaf to the root and removing edge orientations.


For a phylogenetic network $N = (V, E)$, a \emph{covering subtree} of $N$ is a subtree of $N$ with the same leaf-set and root. Recall that a covering subtree $T$ of $N$ is called a support tree of $N$ if $V(T) = V(N)$ holds, and a phylogenetic network $N$ is called tree-based if there exists a support tree of $N$. For simplicity of notation, we denote the set of rooted binary phylogenetic networks by $\mathcal{N}$ and the set of rooted binary tree-based phylogenetic networks by $\mathcal{N}_{\text{TB}}$.

Next, we review definitions from \cite{hayamizu2021structure, suzuki2024bridging} who are relevant to our study. Given a rooted almost-binary phylogenetic network $N$, a connected subgraph $Z$ of $N$ with $m:=|E(Z)|\geq 1$  is a \emph{zig-zag trail} (\emph{in $N$}) if the edges of $Z$ can be permuted as $(e_1, \dots, e_m)$ such that either $\head(e_i)=\head(e_{i+1})$ or $\tail(e_i)=\tail(e_{i+1})$ holds for each $i\in [1,m-1]$. 
A zig-zag trail is represented by an alternating sequence of (not necessarily distinct) vertices and distinct edges, e.g., $(v_0, \allowbreak (v_0, v_1), \allowbreak v_1,\allowbreak  (v_2, v_1), \allowbreak v_2, \allowbreak \dots, \allowbreak (v_m, v_{m-1}), \allowbreak v_m)$, but this can be more concisely written as $v_0 > v_1 < v_2 > \cdots > v_{m-1} < v_m$ or its reverse, where each edge $(v_i, v_{i+1})$ is represented by $v_i > v_{i+1}$. 
	A zig-zag trail $Z$ in $N$ is  \emph{maximal} if $N$ contains no zig-zag trail $Z^\prime$  such that $Z$ is a proper subgraph of $Z^\prime$. Each maximal zig-zag trail in $N$ falls into one of the following four types. 
	A \emph{crown} is a maximal zig-zag trail $Z$ that has even $m:=|E(Z)|\geq 4$ and can be written in the cyclic form $v_0 < v_1 > v_2 < v_3 > \cdots  > v_{m-2} < v_{m-1} > v_{m} = v_0$.  A \emph{fence} is any maximal zig-zag trail that is not a crown. An \emph{N-fence} is a fence $Z$ with odd $m:=|E(Z)|\geq 1$,  expressed as $v_0 > v_1 < v_2 > \cdots < v_{m-1} > v_m$. A fence $Z$ with even $m:=|E(Z)|\geq 2$ is an \emph{M-fence} if expressed as $v_0 < v_1 > v_2 < \cdots < v_{m-1} > v_m$, and it is a \emph{W-fence} if expressed as $v_0 > v_1 < v_2 > \cdots > v_{m-1} < v_m$. 
We will make use of these various types of zig-zag trails because they characterize rooted binary phylogenetic networks:


\begin{theorem}~\cite{hayamizu2021structure}\label{thm::TB}
Let $N\in\mathcal{N}$. Then, there exists a unique decomposition of $N$ into maximal zig-zag trails.
\end{theorem}

We denote the decomposition of $N\in\mathcal{N}$ into maximal zig-zag trails by $\mathcal{Z}(N)$. Furthermore, maximal zig-zag trails provide a certificate of tree-basedness which has been proven in different contexts:

\begin{theorem}~\cite{zhang2016tree,hayamizu2021structure,jetten2016nonbinary}\label{thm::TB::Wfence}
    Let $N\in\mathcal{N}$. Then, $N$ is tree-based if and only if $\mathcal{Z}(N)$ contains no W-fence.
\end{theorem}

\subsection{A review of known tree-basedness proximity measures on rooted phylogenetic networks}

The following proximity measures of tree-basedness are well-known~\cite{francis2018new,van2026proximity}: 
\begin{problem}\label{rootedProx}
    Given a rooted phylogenetic network $N$, calculate the minimum number of
    \begin{enumerate}
    \item leaves in a rooted spanning tree of $N$ which are not leaves in $N$;
    \item vertex disjoint paths that partition $V(N)$ minus the number of leaves in $N$;
    \item leaves added to $N$ to make $N$ tree-based;
    \item vertices in $N$ absent from a covering subtree of $N$;
    \item covering subtrees of $N$ such that every vertex in $N$ is present in at least one of the trees;
    \item edges (and subsequent edge subdivisions) deleted from $N$ to make $N$ tree-based.
\end{enumerate}
\end{problem}
By definition, a given network is tree-based if and only if an optimal solution for \cref{rootedProx}.1 to \cref{rootedProx}.4 and \cref{rootedProx}.5 has objective function value zero and one, respectively. Moreover, \cref{rootedProx}.1 to \cref{rootedProx}.3 are equivalent and can be solved in $\Theta(|E(N)|)$ time~\cite{francis2018new,hayamizu2021structure}. \cref{rootedProx}.4 can be solved in $\mathcal{O}(|V(N)|^2|E(N)|)$ time~\cite{davidov2020maximum}. Interestingly, the authors solve \cref{rootedProx}.4 by studying \cref{rootedProx}.2 under the additional constraints that vertex disjoint paths need to have their endpoint in a leaf. For $n$ leaves, this construction transforms the minimization problem \cref{rootedProx}.2 into a feasibility problem, asking for the existence of a partition of $V(N)$ into vertex disjoint paths with their endpoints in a leaf. Then, the authors model this question as a flow-feasibility problem which they solve by applying the max-flow-min-cut theorem. This approach can be extended to solve \cref{rootedProx}.5 in $\mathcal{O}(|V(N)|^2|E(N)|\log|V(N)|)$ time~\cite{Yoshida2023}. To the best of our knowledge the proof of the last result is part of a master thesis which is neither publicly accessible nor available in English. We did not find any alternative source for the same result. Therefore, for the sake of completion, we include independent proofs of the results in~\cite{Yoshida2023} in~\cref{sec:appendix}. We believe that our proofs are identical or very similar to the inaccessible authors' original. In addition, for rooted not-necessarily binary tree-based networks, there are algorithms to count the number of support trees in $O(|E(N)|^{3/2})$ time~\cite{pons2019tree}. For rooted almost-binary tree-based networks, one can count and enumerate all support trees of $N$ in $O(|E(N)|)$ time and in $O(|E(N)|)$ time delay, respectively \cite{hayamizu2021structure}. Lastly, \cref{rootedProx}.6 is an open problem even if constrained to selected edges whose subsequent removal transforms $N$ into a tree-based network and who, at the time of removal, are not reticulation edges~\cite{van2026proximity}. Moreover, the authors showed that the optimal objective function values of  \cref{rootedProx}.3 and \cref{rootedProx}.6 are incomparable.

We denote the objective function of \cref{rootedProx}.3, \cref{rootedProx}.4 and \cref{rootedProx}.5 by $\delta^*(N)$, $\eta^*(N)$ and $\kappa^*(N)$, respectively. In addition, for a covering subtree $T$ of $N$, we denote the number of vertices in $N$ absent from $T$ by $\eta(N,T)$. We call a covering subtree $T$ of $N$ \emph{maximum} if $\eta(N,T)=\eta^*(N)$. We call $\delta^*(N)$ the \emph{rooted leaf-addition tree-basedness proximity measure} (ATBP), $\eta^*(N)$ the \emph{rooted covering subtree tree-basedness proximity measure} (CTBP) and $\kappa^*(N)$ the \emph{rooted vertex covering tree-basedness proximity measure} (VTBP). For a rooted phylogenetic network $N$ and $W\subseteq V(N)$, we define $\Delta W=\{v\in V(N)\,:\,\exists\,(w,v)\in E(N)\text{ with }w\in W\}$. Then, $\delta^*(N)$ is a supermodular maximization problem:

\begin{restatable}{proposition}{propdeltaSupermodular}\cite{Yoshida2023}\label{prop::deltaSupermodular}
Let $N$ be a rooted phylogenetic network with leaf set $X$. Then,
\begin{align*}
     \delta^*(N)=\max_{W\subseteq V(N)\setminus X}\left(|W|-|\Delta W|\right).
\end{align*}
\end{restatable}
\cref{prop::deltaSupermodular} can be applied to relate the VTBP to the ATBP:

\begin{restatable}{proposition}{propkappaRatio}\cite{Yoshida2023}\label{prop::kappaRatio}
    Let $N$ be a rooted phylogenetic network with leaf set $X$. Then,
    \begin{align*}
        \kappa^*(N)=\max_{W\subseteq V\setminus X}\left\lceil\frac{|W\setminus\Delta W|}{|\Delta W\setminus W|}\right\rceil.
    \end{align*}
    In particular, there exists $W^*\subseteq V\setminus X$ such that
    \begin{align*}
        \kappa^*(N)-1=\left\lceil\frac{\delta^*(N)}{|\Delta W^*\setminus W^*|}\right\rceil,
    \end{align*}
    i.e., $\kappa^*(N)-1\leq\delta^*(N)$. Furthermore, $\kappa^*(N)$ can be calculated in $\mathcal{O}(|V(N)|^2|E(N)|\log|V(N)|)$ time.
\end{restatable}

For binary rooted phlyogenetic networks, there exists an alternative characterization of $\delta^*(N)$ compared to \cref{prop::deltaSupermodular}:

\begin{proposition}[Corollary 5.4 in \cite{hayamizu2021structure}]
    Let $N\in\mathcal{N}$. Then, $\delta^*(N)$ equals the number of maximal $W$-fences in $\mathcal{Z}(N)$.
\end{proposition} 

Let $N$ be a rooted almost-binary phylogenetic network and let $\mathcal{W}_N$ be the set of all W-fences in $\mathcal{Z}(N)$. Let $Z$ and $Z'$ be a M-fence and a W-fence in $\mathcal{Z}(N)$, respectively.  We call $(Z,Z')$ a \emph{M-W-pair} in $N$ if there exists $v\in V(Z)\cap V(Z')$ with $(u,v)\in E(Z)$ and $(v,w)\in E(Z')$. We call a set of M-W-pairs in $N$ with pairwise distinct M-fences and W-fences an \emph{M-W-matching} in $N$. We say an M-W-matching $S$ in $N$ is \emph{$\mathcal{W}_N$-saturated} if $|S| = |\mathcal{W}_N|$.

\begin{theorem}\label{thm:saturated}[Theorems 11 and 19 in~\cite{suzuki2024bridging}]
    Let $N$ be a rooted almost-binary phylogenetic network. Then, $\delta^\ast(N)\leq\eta^\ast(N)$. Equality holds if and only if there exists a $\mathcal{W}_N$-saturated M-W matching in $N$.
\end{theorem}

In total, we conclude from \cref{prop::kappaRatio} and \cref{thm:saturated} that the known tree-basedness proximity measures on rooted almost-binary phylogenetic networks (see \cref{rootedProx}) are related as follows:
\begin{corollary}\label{cor::proximitySummary}
    Let $N$ be a rooted almost-binary phylogenetic network. Then, $\kappa^*(N)-1\leq\delta^*(N)\leq\eta^*(N)$.
\end{corollary}


\begin{figure}[!b]
\centering
\includegraphics[scale=0.4]{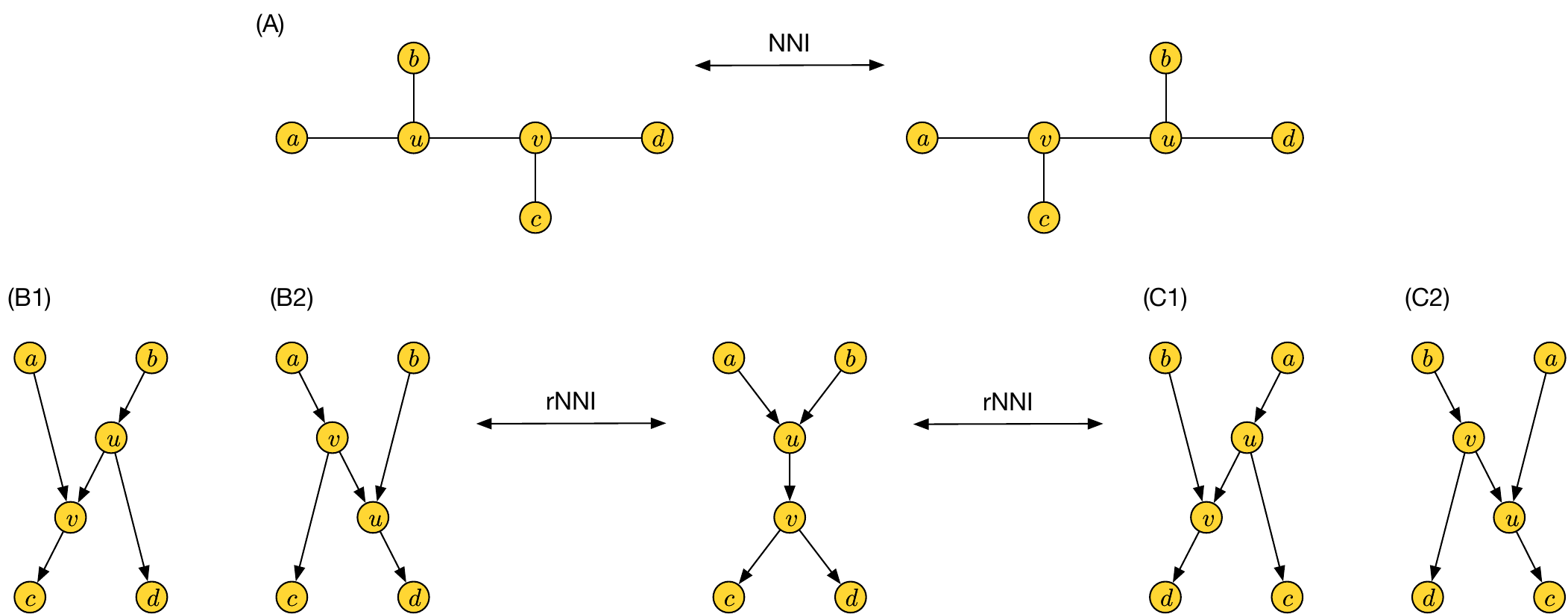}
\caption{In (A) the graph on the right is obtained by a NNI on edge $\{u,v\}$ from the graph on the left and vice versa. In (B1), (B2), (C1) and (C2) the graphs on the left and right are obtained by a rNNI on $(u,v)$ for the graph in the middle and vice versa by maintaining the direction of $(u,v)$ and reversing it, respectively.}\label{fig:NNI}
\end{figure}

\subsection{Topological rearrangement operation of graphs}
The topological rearrangement move NNI shown in Figure~\ref{fig:NNI}(A) has been introduced to the study of both unrooted binary phylogenetic trees~\cite{robinson1971comparison} and networks~\cite{huber2016transforming,huber2016spaces}. 
In the case of rooted binary phylogenetic networks $N$, a rooted version of NNI on an edge $(u,v)\in E(N)$ changes the incidence relations between vertices and edges of $N$ analogously to a NNI with respect to the following properties~\cite{gambette2017rearrangement}:
\begin{enumerate}
    \item there exist pairwise distinct vertices $b$ and $c$ incident to vertices $u$ and $v$, respectively, such that the existence and direction of edges on $\{b,u\},\{c,v\}$ is preserved;
    \item the in- and outdegrees of $u$ and $v$ remain at most two;
    \item the acyclicity of $N$ is preserved.
\end{enumerate}
Furthermore, whenever there is a degree of freedom for the direction of the edge on $\{u,v\}$ after applying a rNNI on $(u,v)$, there exist two distinct rNNIs: one which maintains edge $(u,v)$ and one which reverses its direction. For example, the two rNNIs on $(u,v)$  shown in Figures~\ref{fig:NNI}(B1) and~(B2) only differ in the choice of maintaining edge $(u,v)$ or reversing its direction. This distinction in rNNIs on $(u,v)$ generates symmetries because internal vertices of a phylogenetic network are not labeled. Furthermore, in Figures~\ref{fig:NNI}(B1) and~(B2) property~(1) of the rNNI on $(u,v)$ is satisfied by edges $(b,u)$ and $(v,c)$. In contrast, Figure~\ref{fig:NNI}(C1) and~(C2) shows two rNNIs on $(u,v)$ which satisfy property~(1) for edges $(a,u)$ and $(v,d)$. Hence, while we keep reusing vertices named $a,b,c,d$, the role of vertices $b$ and $c$ in property~(1) is particular and not linked to our naming conventions in later proofs. Notice that the graphs in (B1) and (B2) are equivalent to the graphs in (C2) and~(C1), respectively, up to the relabeling of vertices $u$ and $v$. This means, also the choice of edges in property~(1) generates symmetries. Their direction plays a role, too. Specifically, in the literature a rNNI on $(u,v)$ which does not change the in- and outdegree of all vertices incident to $u$ or $v$ is also called a distance-1 head (tail) move if at least one of the edges on $\{b,u\}$ and $\{c,v\}$ has its head (tail) in $\{u,v\}$~\cite{janssen2021rearranging}. This means rNNIs are a proper generalization of distance-1 head and tail moves in which we allow a change of the in- and outdegree for vertices incident to $u$ and $v$ who do not satisfy property (1) of the rNNI on $(u,v)$.

We write either $\sigma=e_1\to_He_2$ or $\sigma=e_1\to_Te_2$ if $u$ or $v$ is the head or tail of $e_1$ in $N$, respectively, and $e_1$ satisfies property $(1)$ of the rNNI on $(u,v)$ while $e_2$ does not. Similarly, we denote by either $\sigma=e_1\circlearrowleft_H e_2$ or $\sigma=e_1\circlearrowleft_T e_2$ the rNNI operation which functions like $e_1\to_H e_2$ or $e_1\to_T e_2$, respectively, except that the direction on $\{u,v\}$ in $\sigma(N)$ is flipped. For example, for $N$ in the middle and bottom of \cref{fig:NNI}, we obtain the phylogenetic network $\sigma(N)$ in
\begin{description}
    \item[(B1)] for $\sigma=(b,u)\to_H(v,d)=(v,c)\to_T(a,u)$;
    \item[(B2)] for $\sigma=(b,u)\circlearrowleft_H(v,d)=(v,c)\circlearrowleft_T(a,u)$;
    \item[(C1)] for $\sigma=(a,u)\to_H(v,c)=(v,d)\to_T(b,u)$;
    \item[(C2)] for $\sigma=(a,u)\circlearrowleft_H(v,c)=(v,d)\circlearrowleft_T(b,u)$.
\end{description}
Since reversing the direction on $\{u,v\}$ produces the aforementioned symmetries, we can obtain $\sigma(N)=\sigma'(N)$ for pairwise distinct rNNI operations $\sigma$, $\sigma'$. For example, for $\sigma=(b,u)\to_H(v,d)$ and $\sigma'=(a,u)\circlearrowleft_H(v,c)$. A rNNI operation $\sigma$ is well-defined as long as the acylicity of $\sigma(N)$ is ensured, as detailed in Lemma~1 in~\cite{gambette2017rearrangement}. We will not make use of this lemma because we will consider acyclicity as an ancillary condition throughout the article.

We denote the number of rNNI operations to transform $N\in\mathcal{N}$ into a tree-based network $N'\in\mathcal{N}_{\text{TB}}$ by $d_{\text{rNNI}}(N,N')$. Then, $d_{\text{rNNI}}(N,N')$ is well-defined~\cite{francis2018bounds}.
Moreover, we define the central problem of this article, a tree-basedness proximity measure based on rNNI operations:
\begin{problem}\label{problNNI}
    Let $N\in\mathcal{N}$. Then we call
            \begin{align*}
                \delta_{\text{rNNI}}(N)=\min\left\{d_{\text{rNNI}}(N,N')\,:\,N'\in\mathcal{N}_{\text{TB}}\right\}
            \end{align*}
            the \emph{rNNI tree-basedness proximity measure} (NTBP) of $N$.
\end{problem}

\section{Combinatorial relationships between the NTBP and other proximity measures of tree-basedness}\label{sec:rNTBPcomb}
In this section we study the new NTBP by comparing it to the ATBP and CTBP to prove the following variation of~\cref{cor::proximitySummary}:
\begin{theorem}\label{prop::rNNISummary}
    Let $N\in\mathcal{N}$. Then, $\kappa^*(N)-1\leq\delta^*(N)\leq\delta_{\text{rNNI}}(N)\leq\eta^*(N)$.
\end{theorem}
In addition, we show that the NTBP is both pairwise distinct and independent from the ATBP and CTBP. Moreover, we apply rNNI operations to characterize the ATBP, describe the CTBP and solve the NTBP exactly.

\subsection{A characterization of optimal solutions to the ATBP}\label{sec:rNTBPcomb1}
First, we characterize the ATBP of a rooted binary phylogenetic network by necessary and sufficient conditions on rNNI operations. To this end, we broaden the notion of M-W-pairs in a rooted binary phylogenetic network~$N$. Let $Z$ and $Z'$ be a M-fence and a W-fence in $\mathcal{Z}(N)$, respectively. Recall that $(Z,Z')$ is a M-W-pair in $N$ if there exist vertex $v\in V(Z)\cap V(Z')$ and edges $(u,v)\in E(Z)$, $(v,w)\in E(Z')$. We call $v$ the \emph{connecting vertex} of $(Z,Z')$. Then, for some N-fence or crown $Y\in\mathcal{Z}(N)$, we call $(Z,Z')$ a
\begin{itemize}
    \item \emph{distance-1 M-W-pair} in $N$ if there exists an edge $(v_1,v_2)\in E(Y)$ such that $v_1\in V(Z)$, $v_2\in V(Z')$ and there exists no directed path from $v_1$ to $v_2$ of length larger than one in $N$;
    \item \emph{distance-1 W-M-pair} in $N$ if there exists an edge $(v_1,v_2)\in E(Y)$ such that $v_1\in V(Z')$ and $v_2\in V(Z)$;
    \item \emph{$\vee$-sibling M-W-pair} if there exist edges $(v_1,v_3),(v_2,v_3)\in E(Y)$ such that $v_1\in V(Z)$, $v_2\in V(Z')$ and there exists no directed path from $v_1$ to $v_2$;
    \item \emph{$\wedge$-sibling M-W-pair} if there exist edges $(v_3,v_1),(v_3,v_2)\in E(Y)$ such that $v_1\in V(Z)$, $v_2\in V(Z')$ and there exists no directed path from $v_1$ to $v_2$.
\end{itemize}
Observe that the non-existence of a directed path from $v_1$ to $v_2$ in the definition of a distance-1 W-M-pair is ensured by the degree constraints in $N$ and therefore omitted from the definition. We call the edge defining a distance-1 M-W-pair or W-M-pair $(Z,Z')$ the \emph{connecting edge} of $(Z,Z')$.

\begin{theorem}\label{prop:rNNI:1}
    Let $N\in\mathcal{N}$. Then, there exists a rNNI operation $\sigma$ such that $\delta^\ast(\sigma(N))=\delta^\ast(N)-1$ if and only if at least one of the following conditions holds: in $N$ there exists a
        \begin{enumerate}
            \item M-W-pair;
            \item distance-1 M-W-pair $(Y,Z)$ such that its connecting edge is incident to the endpoint of $Z$;
            \item distance-1 W-M-pair;
            \item $\vee$-sibling M-W-pair;
            \item $\wedge$-sibling M-W-pair.
        \end{enumerate}
    In addition, there exists no rNNI operation $\sigma$ in~$N$ such that $\delta^\ast(\sigma(N))\leq\delta^\ast(N)-2$.
\end{theorem}

\begin{figure}[!b]
\centering
\includegraphics[width=0.8\linewidth]{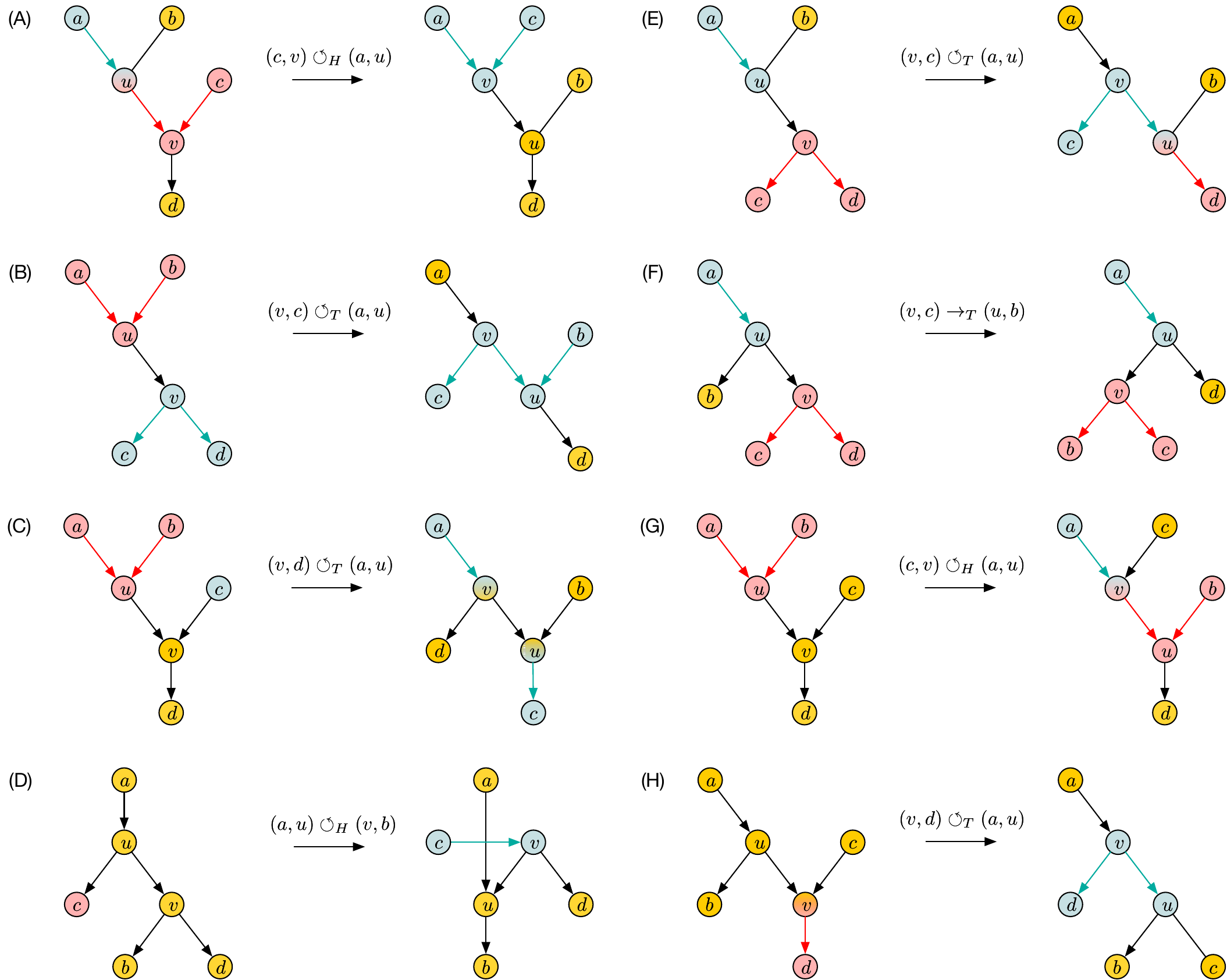}
\caption{Seven rNNI operations for the respective graphs shown under the assumption that all undirected edges are assigned a direction consistent with the definition of the rNNI operations. Colors refer to maximal zig-zag trails in the proof of Proposition~\ref{prop:rNNI:1}.}\label{fig:rNNI_1}
\end{figure}

\begin{proof}
    First, we consider the if-direction:\\
    Condition~1: let $(Y,Z)$ be a M-W-pair in $N$ with connecting vertex $u$. Then, the graph on the left in Figure~\ref{fig:rNNI_1}(A) shows the intersection of a M-fence $Y$ (blue) and a W-fence $Z$ (red) in $u$. Observe that $b$ is either in $V(Y)$ or $V(Z)$ depending on the direction of edge $\{b,u\}$. 
    \begin{description}
        \item[Case 1:] $(b,u)\in E(N)$. Consider the rNNI operation $\sigma=(c,v)\circlearrowleft_H(a,u)$ as shown in Figure~\ref{fig:rNNI_1}(A). Notice that applying $\sigma$ to $N$ creates a cycle only if $(u,b)\in E(N)$ and there exists a directed path from $b$ to $a$ or $c$ in $N$. Hence, $\sigma$ is well-defined. Then, for $(b,b')\in E(N)$, $b'\neq u$, we have $b,b'\in V(Y)$ because of the maximality of $Y$. Hence, $b$ is in a N-fence in $\mathcal{Z}(\sigma(N))$.
        Since $(a,u)$ and $(c,v)$ are in a M-fence and W-fence in $\mathcal{Z}(N)$, respectively, edges $(a,v),(c,v)$ are in a crown or N-fence in $\mathcal{Z}(\sigma(N))$. In addition, $(v,u)$ forms a maximal N-fence in $\mathcal{Z}(\sigma(N))$, too. Therefore, in total $|\mathcal{W}_{\sigma(N)}|=|\mathcal{W}_{N}|-1$. In other words, $\delta^\ast(\sigma(N))=\delta^\ast(N)-1$.
        \item[Case 2:] $(u,b)\in E(N)$. Then, $(u,b)\in V(Z)$ because of the maximality of $Z$. This means, $b$ is in a W-fence in $\mathcal{Z}(\sigma(N))$ if and only if $d$ is in a W-fence in $\mathcal{Z}(N)$. Therefore, if $\sigma$ from Case~1 preserves the acyclicity of $N$, then we can draw the same conclusion as in Case~1. Otherwise, there exists a directed path from $b$ to $c$ in $N$. Since $(u,b)\in E(Z)$ and $Z$ is a W-fence, there exists an edge $(b',b)\in E(Z)$, $b'\neq u$. Thus, we can substitute $v$ by $b$ to apply Case~1 again.
    \end{description}    
    Condition~2: let $(Y,Z)$ be a distance-1 M-W-pair in $N$ with connecting edge $(u,v)$ such that $u$ and $v$ are the endpoints of $Y$ and $Z$, respectively. In this case the graph on the left in Figure~\ref{fig:rNNI_1}(A) shows all edges incident to $(u,v)$ when considering $(v,d)\in E(Z)$ instead of $(u,v),(c,v)\in E(Z)$. Thus, our claim follows from analogous arguments as under condition~1.
    \\~\\
    Condition~3: let $(Y,Z)$ be a distance-1 W-M-pair in $N$ with connecting edge $(u,v)$. Observe that $u$ and $v$ are not at the end of $Y$ and $Z$, respectively, because $N$ is binary. Hence, all edges incident to $(u,v)$ are shown on the left in Figure~\ref{fig:rNNI_1}(B). The edges in $Y$ and $Z$ are depicted in red and blue, respectively. Then, the rNNI operation $\sigma$ in (B) puts edges $(v,c),(v,u),(b,u)$ in a crown or N-fence. Notice that $\sigma$ can not create a cycle by construction. In addition, edges $(a,v)$ and $(u,d)$ can not be in W-fences in $\mathcal{Z}(\sigma(N))$ by definition of $Y$ and $Z$. Thus, $\delta^\ast(\sigma(N))=\delta^\ast(N)-1$.
    \\~\\
    Condition~4: let $(Y,Z)$ be a $\vee$-sibling M-W-pair in $N$. Then, there exist edges $(u,v),(c,v)\in E(Y')$ for a N-fence $Y'$ such that $c\in V(Y)$ and $u\in V(Z)$. Then, all edges incident to $(u,v)$ are shown on the left in Figure~\ref{fig:rNNI_1}(C) (vertices and edges in $Y$ and $Z$ in blue and red, respectively). We have $d_{\text{out}}(c)=2$ because $Y'$ contains $(c,v)$ and is a N-fence in $\mathcal{Z}(N)$. This means, $c$ is an endpoint of $Y$. Hence, $\sigma=(v,d)\circlearrowleft_T(a,u)$ is well-defined because there exists no directed path from $c$ to $b$ in $N$ by our choice of $(Y,Z)$. $\sigma$ ensures that edges $(a,v),(u,c)\in E(\sigma(N))$ are not in a W-fence in $\mathcal{Z}(\sigma(N))$. Furthermore, $(v,d)$ is in a W-fence in~$\mathcal{Z}(N)$ if and only if edges $(v,d),(v,u),(b,u)$ are in a W-fence in~$\mathcal{Z}(\sigma(N))$. Thus, $\delta^\ast(\sigma(N))=\delta^\ast(N)-1$.
    \\~\\
    Condition~5: let $(Y,Z)$ be a $\wedge$-sibling M-W-pair in $N$. Then, there exist edges $(u,v),(u,c)\in E(Y')$ for a N-fence $Y'$ such that $v\in V(Y)$ and $c\in V(Z)$. Then, all edges incident to $(u,v)$ are shown on the left in Figure~\ref{fig:rNNI_1}(D). We have $d_{\text{in}}(c)=2$ because $Y'$ contains $(u,c)$ and is a N-fence in $\mathcal{Z}(N)$. Therefore, $c$ is an endpoint of $Z$. Then, $\sigma=(a,u)\circlearrowleft_H(v,b)$ is well-defined because there exists no directed path from $b$ to $c$ in $N$ by our choice of $(Y,Z)$. $\sigma$ ensures that edge $(c,v)$ is in a N-fence in $\mathcal{Z}(\sigma(N))$. Since $(v,d),(v,b)\in E(Y)$, edges $(a,u),(v,u),(v,d)$ are not in a W-fence in~$\mathcal{Z}(\sigma(N))$ and $(u,b)$ is in a N-fence in~$\mathcal{Z}(\sigma(N))$. Thus, $\delta^\ast(\sigma(N))=\delta^\ast(N)-1$.
    \\~\\
    Next, we prove the only-if-direction by contraposition: assume conditions~1 to~5 do not hold. Moreover, we do not need to distinguish the existence of directed paths from conditions~1 to~5 because their only purpose is to ensure the acyclicity of rNNI operations. Let $(u,v)\in E(N)$. Then, we consider all possible rNNI operations on $(u,v)$ which are well-defined after excluding conditions~1 to~5.
    \begin{description}
        \item[Case 1:] There exists a distance-1 M-W-pair $(Y,Z)$ with connecting edge $(u,v)$ such that $v$ is not an endpoint of $Z$. Then, for $v$ not at the end of $Z$, the graph on the left in Figure~\ref{fig:rNNI_1}(E) shows all edges incident to $(u,v)$, edges in $Y$ in blue and edges in $Z$ in red. Here the edge $e$ on $\{b,u\}$ determines whether $e$ is in $E(Y)$ or not. In either case the rNNI operation $\sigma=(v,c)\circlearrowleft_T(a,u)$ yields a N-fence or crown in $\mathcal{Z}(\sigma(N))$ shown in blue on the right because $(a,u)$ is in a M-fence in $\mathcal{Z}(N)$. 
        \begin{description}
            \item[Case 1.1:] $(b,u)\in E(N)$. Then, $(u,d)$ is in a W-fence in $\mathcal{Z}(\sigma(N))$. Hence, $\delta^*(\sigma(N))=\delta^*(N)$. Notice that $\sigma$ is the only rNNI operation on $(u,v)$ up to symmetries.
            \item[Case 1.2:] $(u,b)\in E(N)$. Since we assume that $N$ contains no M-W-pair, $(u,b)\in E(N)$ is in a N-fence or crown in $\mathcal{Z}(N)$. Therefore, $(u,b),(u,d)$ are in a W-fence in $\mathcal{Z}(\sigma(N))$, i.e., $\delta^*(\sigma(N))=\delta^*(N)$. Then, the only remaining rNNI operation $\sigma$ on $(u,v)$ up to symmetries is $\sigma'=(v,c)\to_T(u,b)$ (shown in Figure~\ref{fig:rNNI_1}(F)). In this case, since $(u,b)$ is in a N-fence or crown in~$\mathcal{Z}(N)$ and $(v,c)\in E(Z)$, $(v,b),(v,c)$ are in a W-fence in~$\mathcal{Z}(\sigma'(N))$. Also, since $(v,d)\in E(Z)$, $(u,d)$ is in a N-fence in~$\mathcal{Z}(\sigma'(N))$. Therefore, in total $\delta^*(\sigma'(N))=\delta^*(N)$.
        \end{description}
        \item[Case 2:] $u\in V(Y)$ for a M-fence $Y\in\mathcal{Z}(N)$ and $v\notin V(Z)$ for all W-fences $Z\in\mathcal{Z}(N)$. Hence, we have an empty intersection between any edge incident to $(u,v)$ and any W-fence in~$\mathcal{Z}(N)$. In other words, $\delta^*(\sigma(N))=\delta^*(N)$ for all rNNI operations $\sigma$ on $(u,v)$.
        \item[Case 3:] $u\in V(Y)$ for a W-fence $Y\in\mathcal{Z}(N)$ and $v\notin V(Z)$ for all M-fences $Z\in\mathcal{Z}(N)$.
        \begin{description}
            \item[Case 3.1:] $d_{\text{out}}(v)=2$. Then, the only applicable rNNI operation $\sigma$ on $(u,v)$ up to symmetries is shown in Figure~\ref{fig:rNNI_1}(B). Observe that vertex $c$ or $d$ or both have in-degree two because $v$ is not in any M-fence in~$N$. Hence, $(v,c)$ or $(u,d)$ or both are in a W-fence in $\mathcal{Z}(\sigma(N))$ using the fact that $(b,u)\in E(Y)$. Thus, $\delta^*(\sigma(N))\geq\delta^*(N)$.
            
            \item[Case 3.2:] $d_{\text{in}}(u)=d_{\text{in}}(v)=2$. Figure~\ref{fig:rNNI_1}(C) and (G) depict the neighborhood of $(u,v)$. First, assume vertices and edges in $Y$ are shown in red.
            
            If $d_{\text{in}}(c)=1$, then the rNNI operation $\sigma$ shown in Figure~\ref{fig:rNNI_1}(C) is well-defined. Observe that $(u,c)$ is in a W-fence in $\mathcal{Z}(\sigma(N))$. Otherwise, condition~4 would hold. Hence, with analogous arguments to our proof of the if-direction under condition~4 it follows that $\delta^*(N)\geq\delta^*(\sigma(N))$. The only other rNNI operation up to symmetries is shown in Figure~\ref{fig:rNNI_1}(G). Since $u$ is in at least one W-fence in $\mathcal{Z}(N)$ and in no M-fence in $\mathcal{Z}(N)$, its follows that $\delta^*(\sigma(N))=\delta^*(N)$.
            
            If $d_{\text{in}}(c)=2$, then $\sigma=(v,d)\circlearrowleft_T(a,u)$ is not well-defined and our arguments for the rNNI operation in Figure~\ref{fig:rNNI_1}(G) are analogous. Thus, $\delta^*(\sigma(N))=\delta^*(N)$.

            Finally, if we violate the assumption that vertices and edges in $Y$ are shown in red, then $(u,v)\in E(Y)$. Hence, the arguments for both cases $d_{\text{in}}(c)=1$ and $d_{\text{in}}(c)=2$ still apply.

            \item[Case 3.3:] $d_{\text{in}}(u)=1$ and $d_{\text{in}}(v)=2$. Figure~\ref{fig:rNNI_1}(H) depicts the neighborhood of $(u,v)$. Then, $(u,v)$ is in a W-fence in $\mathcal{Z}(N)$ because $u\in V(Y)$. Therefore, no vertex adjacent to $u$ or $v$ is in a M-fence in $\mathcal{Z}(N)$ because $N$ contains no M-W-pair and $v$ is not in any M-fence in $\mathcal{Z}(N)$. 

            First, consider rNNI operation $\sigma=(c,v)\circlearrowleft_H(a,u)$ or $\sigma=(c,v)\to_H(v,b)$. In both cases $(c,v)$ is in a W-fence in $\mathcal{Z}(\sigma(N))$ because $(c,v)$ is in a W-fence in $\mathcal{Z}(N)$ and no vertex adjacent to $u$ is in a M-fence in $\mathcal{Z}(N)$. The only other rNNI operation on $(u,v)$ up to symmetries is $\sigma =(v,d)\circlearrowleft_T(a,u)$ (shown in Figure~\ref{fig:rNNI_1}(H)). 
            
            Then, for $(c,u)\in E(\sigma(N))$, the maximal zig-zag trail containing edge $(v,d)$ does not change its type from $N$ to $\sigma(N)$ because no vertex adjacent to $v$ is in a M-fence in $\mathcal{Z}(N)$. In addition, $(u,b)$ is in W-fence in $\mathcal{Z}(\sigma(N))$. 
            
            For $(u,c)\in E(\sigma(N))$, edges $(v,d),(v,u)$ are in a N-fence or M-fence in $\mathcal{Z}(\sigma(N))$ because no vertex adjacent to $v$ is in a M-fence in $\mathcal{Z}(N)$. Since $(c,v)$ is in a W-fence in $\mathcal{Z}(N)$ and $d_{\text{in}}(c)=1$ for $\sigma$ to be well-defined, we deduce that vertex $c$ is in two W-fences in $\sigma(N)$. 
            
            In total we conclude that $\delta^*(\sigma(N))\geq\delta^*(N)$. 
        \end{description}
        \item[Case 4:] both $u$ and $v$ are neither in a M-fence nor W-fence in~$\mathcal{Z}(N)$. Then, the edges incident to $(u,v)$ have an empty intersection with all edges in W-fences in~$\mathcal{Z}(N)$. However, it might be possible to apply the rNNI operation $\sigma=(a,u)\circlearrowleft_H(v,b)$ shown in Figure~\ref{fig:rNNI_1}(D) to transform a W-fence in $\mathcal{Z}(N)$ with endpoint $c$ into a N-fence with endpoint $v$ in $\mathcal{Z}(\sigma(N))$. Observe that this case is the only one in which a vertex $c$ of a W-fence is adjacent to $u$ or $v$ and $c$ does not appear in a W-fence in $\mathcal{Z}(\sigma(N))$. We argue that $\delta^*(\sigma(N))=\delta^*(N)$: since edges $(a,u)$ and $(v,d)$ are in N-fences or crowns in~$\mathcal{Z}(N)$, it is not in a W-fence in $\mathcal{Z}(\sigma(N))$ only if $d_{\text{in}}(d)=1$. However, $d_{\text{in}}(d)=1$ if and only if $(u,b)$ is in a W-fence in $\mathcal{Z}(\sigma(N))$ because $(v,b)$ is in a N-fence or crown in $\mathcal{Z}(N)$. Thus, no rNNI operation on $(u,v)$ can decrease $\delta^*(N)$.
        \item[Case 5:] $u$ is neither in a M-fence nor a W-fence in~$\mathcal{Z}(N)$, $v\in V(Z)$ for a W-fence $Z\in\mathcal{Z}(N)$ and $d_{\text{in}}(v)=1$. Then, the orientation of edges incident to $(u,v)$ can be derived from the graph on the left in Figure~\ref{fig:rNNI_1}(E). In particular, $(v,c),(v,d)\in E(Z)$. We consider the rNNI operations in Figures~\ref{fig:rNNI_1}(E) and (F) which constitute all rNNI operations on $(u,v)$ up to symmetries. Then, analogous arguments to Case~1 apply to conclude that $\delta^*(N)$ cannot decrease.
        \item[Case 6:] $u$ is neither in a M-fence nor a W-fence in~$\mathcal{Z}(N)$, $v\in V(Z)$ for a W-fence $Z\in\mathcal{Z}(N)$ and $d_{\text{in}}(v)=2$. Then, the orientation of edges incident to $(u,v)$ can be derived from the graph on the left in Figure~\ref{fig:rNNI_1}(A) such that only edge $(v,d)$ is in a W-fence in~$\mathcal{Z}(N)$.
        \begin{description}
            \item[Case 6.1:] $(b,u)\in E(N)$. Then, all rNNI operations $\sigma$ up to symmetries are shown in Figures~\ref{fig:rNNI_1}(C) and~\ref{fig:rNNI_1}(G). In the latter case, since both $(b,u)$ and $(c,v)$ are in a N-fence or crown, $\delta^*(\sigma(N))=\delta^*(N)$. In the former case assume that $c$ is not in a M-fence in~$\mathcal{Z}(N)$. Otherwise, condition~4 holds. This means, $(u,c)$ is in a W-fence in $\mathcal{Z}(\sigma(N))$. Thus, $\delta^*(\sigma(N))=\delta^*(N)$.
            \item[Case 6.2:] $(u,b)\in E(N)$. Then, the graph on the left in Figure~\ref{fig:rNNI_1}(H) depicts the neighborhood of $(u,v)$ with $Z$ shown in red. Hence, for rNNI operation $\sigma=(c,v)\circlearrowleft_H(a,u)$ (see Figure~\ref{fig:rNNI_1}(A)), either $(c,v)$ or $(u,b)$ is in a W-fence in~$\mathcal{Z}(\sigma(N))$. Therefore, $\delta^*(\sigma(N))\geq\delta^*(N)$. Otherwise, the only remaining rNNI operation $\sigma$ up to symmetries is shown in Figure~\ref{fig:rNNI_1}(H).
            
            If $(u,c)\in E(\sigma(N))$, then $d_{\text{in}}(c)=1$ in $N$ for $\sigma$ to be well-defined. This means, $d_{\text{in}}(b)=2$ because $(c,v)$ is not in a M-fence in $\mathcal{Z}(N)$. Hence, either edges $(u,b),(u,c)$ are in W-fence in $\mathcal{Z}(N)$ or $c$ is the endpoint of a W-fence in $\mathcal{Z}(N)$. Thus, $\delta^*(\sigma(N))\geq\delta^*(N)$.

            If $(c,u)\in E(\sigma(N))$, then either $(v,b)$ is in a W-fence in $\mathcal{Z}(\sigma(N))$ or $(v,d)\in E(Z)$ implies that $(c,u)$ is in a W-fence in $\mathcal{Z}(\sigma(N))$. In both cases $\delta^*(\sigma(N))\geq\delta^*(N)$.
        \end{description}
    \end{description}
    We conclude that for all $(u,v)\in E(N)$ there exists no rNNI operation $\sigma$ on $(u,v)$ such that $\delta^*(\sigma(N))\leq\delta^*(N)-1$ as long as conditions~1 to~5 are not satisfied.
\end{proof}

As a consequence of \cref{prop:rNNI:1} we can observe that if $N\in\mathcal{N}$ has a $\mathcal{W}_N$-saturated M-W-matching, then $\delta^*(N)=\delta_{\text{rNNI}}(N)$ due to \cref{prop:rNNI:1}.1. We generalize such matchings as follows: we call a set of disjoint collections $\{Z,Z'\}$ for M-fences $Z$ and W-fences $Z'$ in $\mathcal{Z}(N)$, each of which fits into one of \cref{prop:rNNI:1}.1 to~\cref{prop:rNNI:1}.5, a \emph{M-W-rNNI-matching} in $N$. We say a M-W-rNNI-matching $S$ in $N$ is $\mathcal{W}_N$-saturated if $|S|=|\mathcal{W}_N|$. Then, we draw the following conclusion from \cref{prop:rNNI:1}:

\begin{corollary}\label{cor:rNNI:1}
    Let $N\in\mathcal{N}$. Then, $\delta^*(N)\leq\delta_{\text{rNNI}}(N)$. Equality holds if and only if, for $k=\delta_{\text{rNNI}}(N)$,
    \begin{enumerate}
        \item there exists a sequence $\sigma_1, \dots, \sigma_{k}$ such that $\sigma_i$ is a rNNI operation for $N_i=\sigma_{i-1}(\sigma_{i-2}(\dots (N)\dots))$ for all $i\in\{1,\dots,k\}$;
        \item there exists a $\mathcal{W}_{N_i}$-saturated M-W-rNNI-matching in $N_i$ for all $i\in\{1,\dots,k\}$.
    \end{enumerate}
\end{corollary}

\begin{figure}[!t]
\centering
\includegraphics[width=0.5\linewidth]{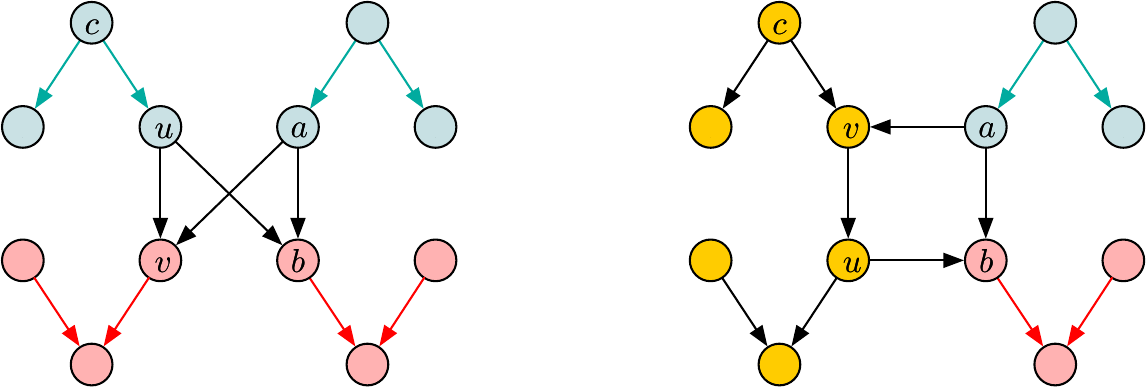}
\caption{On the left a subgraph of a network $N\in\mathcal{N}$. On the right a subgraph of $\sigma(N)$ for $\sigma=(a,v)\circlearrowleft_H(c,u)$.}\label{fig:crown}
\end{figure}

\cref{fig:crown} illustrates why the first condition in \cref{cor:rNNI:1} is necessary:
\begin{example}\label{exmp:rNNI:1}
    Consider $N\in\mathcal{N}$ having the graph on the left in \cref{fig:crown} as its subgraph. Then, edges $(u,v)$ and $(a,b)$ are the connecting edges of pairwise distinct distance-1 M-W-pairs, respectively. Since $v$ and $b$ are the endpoint of a W-fence in $\mathcal{Z}(N)$, \cref{prop:rNNI:1}.2 holds for both distance-1 M-W-pairs. However, after applying rNNI operation $\sigma=(a,v)\circlearrowleft_H(c,u)$ to $N$ neither tuple of a M-fence and W-fence in $\mathcal{Z}(N)$ is a distance-1 M-W-pair in $\sigma(N)$. One because $(v,u)$ is not a connecting edge and the other because there exists a directed path from $a$ to $b$ which does not contain connecting edge $(a,b)$. This means, while there might exist a $W_{N}$-saturated M-W-rNNI matching in $N$, there cannot exist a $W_{\sigma(N)}$-saturated M-W-rNNI matching in $\sigma(N)$.
\end{example}

While \cref{cor:rNNI:1} characterizes all rooted binary phylogenetic networks on which the optimal objective function values of the ATBP and NTBP coincide, we can also identify a family of networks on which the two proximity measures are arbitrarily different:

\begin{figure}[!t]
    \centering
    \includegraphics[width=0.65\linewidth]{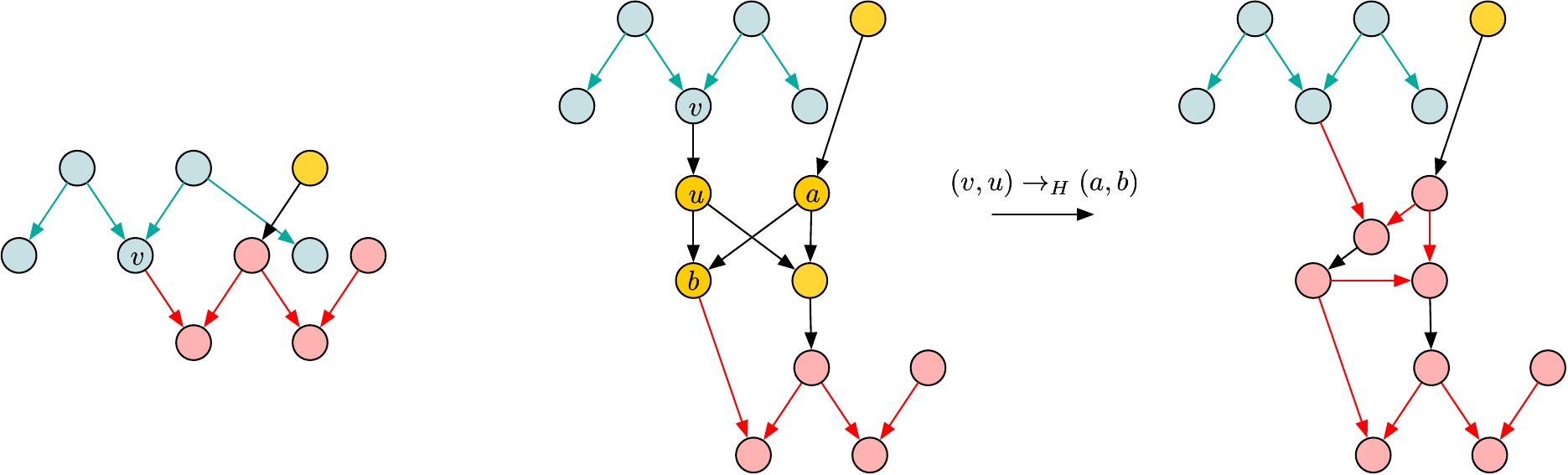}
    \caption{Three graphs used in the proof of \cref{prop::delta+k}.}
    \label{fig:delta+k}
\end{figure}

\begin{restatable}{proposition}{propdeltak}\label{prop::delta+k}
    For any integers $k>0$ and any $N\in\mathcal{N}$ for which there exists a $\mathcal{W}_N$-saturated M-W-rNNI-matching, there exists $N'\in\mathcal{N}$ such that $\mathcal{W}_{N'}=\mathcal{W}_N$, $\mathcal{M}_{N'}=\mathcal{M}_N$ and $\delta^*(N')+k=\delta_{\text{rNNI}}(N')$.
\end{restatable}
\begin{proof}
    First, assume there exists a $\mathcal{W}_N$-saturated M-W-rNNI-matching in $N$ for which \cref{prop:rNNI:1}.1 is satisfied at least once. This means, there exists a M-W-pair $(Z,Z')$ in $N$. Let $v\in V(Z)\cap V(Z')$. The graph on the left in \cref{fig:delta+k} shows the M-fence $Z$ in blue and W-fence $Z'$ in red in the neighborhood of vertex $v$. Then, we can insert a crown and two N-fences as shown in the graph in the middle of \cref{fig:delta+k} to construct $N'\in\mathcal{N}$ with $\mathcal{W}_{N'}=\mathcal{W}_N$ and $\mathcal{M}_{N'}=\mathcal{M}_N$ but $(Z,Z')$ is not a M-W-pair in $N'$. Indeed, the insertion of the crown and N-fences increases the shortest topological distance between $v$ and a vertex in $Z'$ by two. Then, the rNNI operation $\sigma=(v,u)\to_H(a,b)$ yields the network $\sigma(N')$ depicted on the right in \cref{fig:delta+k}. Clearly, there exists a $\mathcal{W}_{\sigma(N')}$-saturated M-W-rNNI-matching in $\sigma(N')$ because there exists a $\mathcal{W}_{N}$-saturated M-W-rNNI-matching in $N$ and both matchings differ only in the transformation of the W-fence $Z'$. Therefore, \cref{cor:rNNI:1} implies $\delta^*(N')+1=\delta_{\text{rNNI}}(N')$. We can repeat the arguments by adding more crowns and N-fences on the shortest path between vertices in $Z$ and $Z'$ to construct $N'\in\mathcal{N}$ such that $\mathcal{W}_{N'}=\mathcal{W}_N$, $\mathcal{M}_{N'}=\mathcal{M}_N$ and $\delta^*(N')+k=\delta_{\text{rNNI}}(N')$. The cases in which any of \cref{prop:rNNI:1}.2 to~\cref{prop:rNNI:1}.5 is satisfied at least once are analogous (see~\cref{sec:appendix}).
\end{proof}

\begin{figure}[!b]
    \centering
    \includegraphics[width=0.5\linewidth]{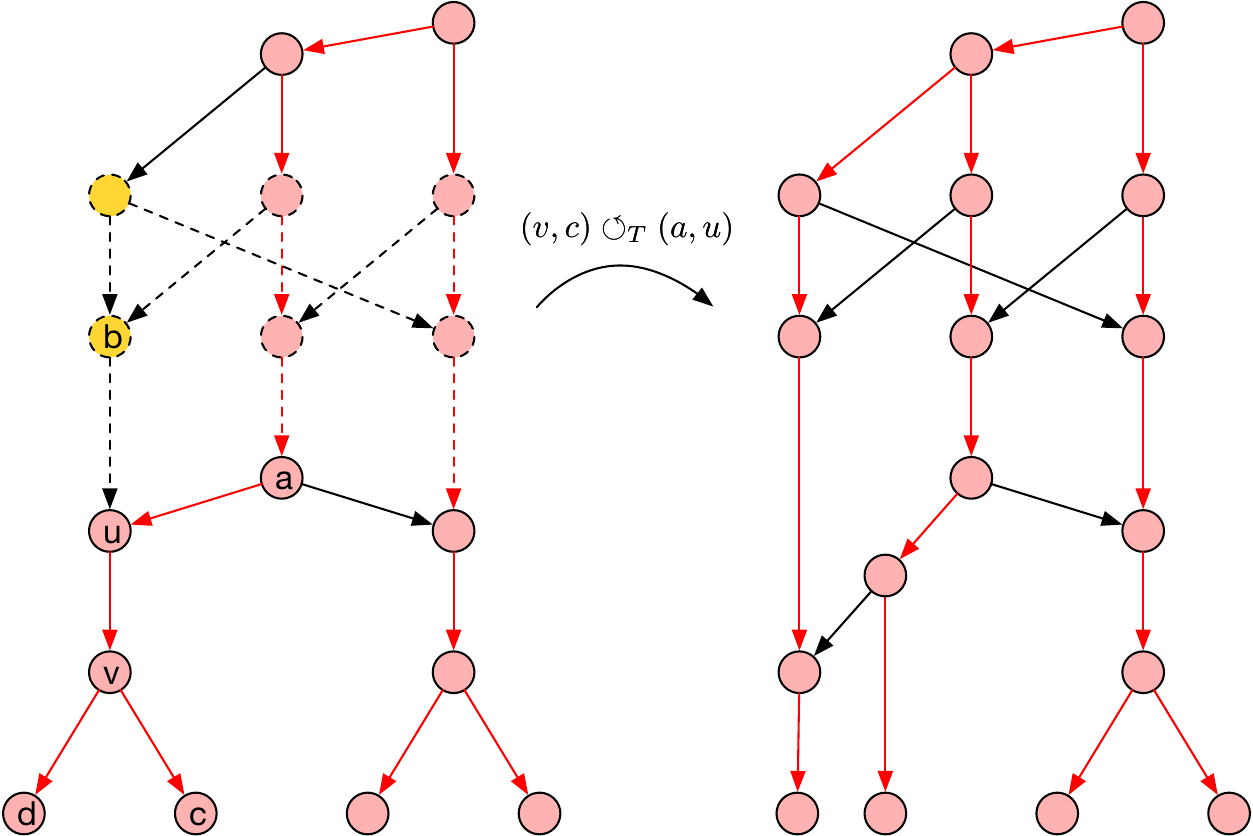}
    \caption{Two rooted binary phylogenetic networks which differ by the rNNI operations shown. A maximum covering subtree for each network is shown in red.
    }
    \label{fig:covering-failure}
\end{figure}

Since not every $\mathcal{W}_N$-saturated M-W-rNNI-matching is a $\mathcal{W}_N$-saturated M-W-matching, we know from \cref{thm:saturated} and \cref{cor:rNNI:1} that $\delta^*(N)=\delta_{\text{rNNI}}(N)<\eta^*(N)$ can occur. We will study this relation in the next subsection and also investigate the conclusion drawn from \cref{thm:saturated} and \cref{prop::delta+k} that there exist rooted binary phylogenetic networks $N'$ with $\delta^*(N')+k=\delta_{\text{rNNI}}(N')\leq\eta^*(N')$.

\subsection{On covering subtrees and the CTBP}\label{sec:rNTBPcomb2}
For $N\in\mathcal{N}$, we begin our study of covering subtrees by observing that the gap between $\delta_{\text{rNNI}}(N)$ and $\eta^*(N)$ can be arbitrarily large:
\begin{example}\label{exmp::eta}
Consider the networks $N_1$ and $N_2$ on the left and right in \cref{fig:covering-failure}, respectively. Observe that $N_1$ is not tree-based but $N_2$ is tree-based because one can verify that $\eta^*(N_1)=2$ and $\eta^*(N_2)=0$. Furthermore, $N_1$ has a $\mathcal{W}_{N_1}$-saturated M-W-rNNI-matching because $|\mathcal{W}_{N_1}|=1$ and edge $(u,v)$ is the connecting edge of a distance-1 W-M-pair (see \cref{prop:rNNI:1}). Thus, we conclude from \cref{cor:rNNI:1} that $1=\delta^*(N_1)=\delta_{\text{rNNI}}(N_1)<\eta^*(N_1)=2$. The crown and N-fences indicated by dashed lines in \cref{fig:covering-failure} can be duplicated multiple times to increase the length of the directed paths from the root to the leaves to any arbitrary but fixed length. Hence, the number of yellow vertices $\eta^\ast(N_1)$ can be made arbitrarily large while $\delta_{\text{rNNI}}(N_1)=1$ holds.
\end{example}

The phenomenon in \cref{exmp::eta} can be quantified for general rooted binary phylogenetic networks as follows:

\begin{lemma}\label{lem::switching}
    Let $N\in\mathcal{N}$. For a maximum covering subtree $T$ of $N$, let $(u,v)\in E(T)$ for a reticulation vertex~$u$ and an internal tree vertex $v$ in $N$. Let $\mathcal{P}$ be the set of directed paths with startpoint in $V(T)$, endpoint $u$ and $E(P)\subset E(N)\setminus E(T)$.
    \begin{enumerate}
        \item Assume $E(T)$ contains at least two edges incident to $u$ and all edges incident to $v$. Then, for all rNNI operations~$\sigma$ on $(u,v)$ we have
        \begin{align*}
            \eta^*(\sigma(N))=\eta^*(N)-\max\{|E(P)|\,:\,P\in\mathcal{P}\}+1.
        \end{align*}
        \item Assume $E(T)$ contains exactly two edges adjacent to $(u,v)$. Then, for all rNNI  operations $\sigma=e_1\square e_2$ on $(u,v)$ we have
        \begin{align*}
            \eta^*(\sigma(N))\leq\eta^*(N)+1
        \end{align*}
        with equality if and only if $e_1,e_2\in E(T)$ or $e_1,e_2\in E(N)\setminus E(T)$, and there exists a maximum covering subtree $T_{\sigma}$ of $\sigma(N)$ such that $(u,v),(v,u)\notin  E(T_{\sigma})$ and either $u\in V(T_{\sigma})$ or $v\in V(T_{\sigma})$.
    \end{enumerate}
\end{lemma}

\begin{figure}[!t]
    \centering
    \includegraphics[width=0.5\linewidth]{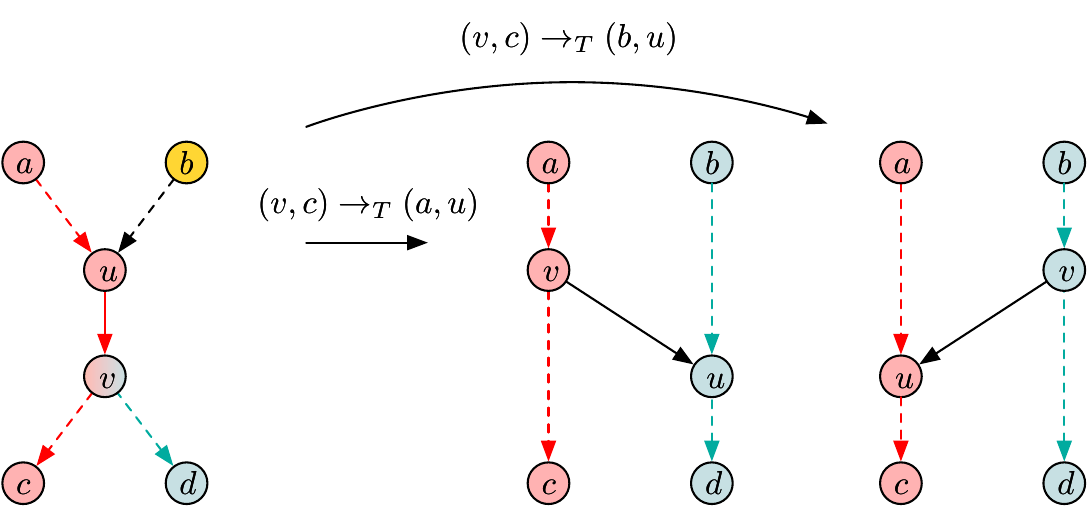}
    \caption{The graphs on the right can be obtained from the graphs on the left applying the respective rNNI operations. Red and blue vertices and edges indicate covering subtrees in the proof of \cref{lem::switching}. Dashed edges might be subdivided a finite number of times.
    }
    \label{fig:switching}
\end{figure}

\begin{proof}
    Notice that $\mathcal{P}$ is not empty because $T$ is a tree. We label vertices and edges adjacent and incident to $(u,v)$, respectively, as shown on the left in \cref{fig:switching} based on our assumptions for $u$ and $v$. Consider a maximum covering subtree $T$ of $N$ containing edge $(u,v)$. Then, $(v,c)\in E(T)$ or $(v,d)\in E(T)$ because $v$ is not a leaf of $N$ and $(a,u)\in E(T)$ or $(b,u)\in E(T)$ because $u$ is not the root of $N$. Let $\sigma=(v,c)\to_T(a,u)$ and $\sigma'=(v,d)\to_T(b,u)$. Then, the edges incident to $u$ and $v$ in $\sigma(N)$ and $\sigma'(N)$ are shown in the middle and on the right in \cref{fig:switching}, respectively. Observe that for all other rNNI operations $\sigma''$ on $(u,v)$ involving edges incident to $u$ or $v$ in $N$ there exists a graph isomorphism between $\sigma''(N)$ and either $\sigma(N)$ or $\sigma'(N)$ up to a relabeling of vertices $a,b,c,d$. Specifically, by exchanging the role of $a$ and $b$ and/or the roles of $c$ and $d$.

    First, assume $(a,u),(v,c),(v,d)\in E(T)$ and $(b,u)\notin E(T)$, i.e., $E(T)$ contains exactly two edges indicent to $u$ and all edges incident to $v$ ($T$ is depicted by red and blue vertices and edges on the left in \cref{fig:switching}). Since $T$ contains the root, we can generalize the restriction $(b,u)\notin E(T)$ by relabeling a vertex in $T$ as $b$ for which there exists a directed path $P(b,u)\in\mathcal{P}$. Then, red and blue vertices and edges of the graph in the middle and on the right of \cref{fig:switching} are in a maximum covering subtree $T_{\sigma}$ of $\sigma(N)$ and $T_{\sigma'}$ of $\sigma'(N)$, respectively, because of the maximality of $T$. Hence,
    \begin{align*}
        \eta^*(\sigma(N))=\eta^*(\sigma'(N))=\eta^*(N)-
        \left(|P(b,u)|-1\right).
    \end{align*}
    The same conclusions follows if we exchange the roles of $a$ and $b$ and/or the roles of $c$ and $d$. Thus, our first claim holds.
    
    Next, assume $(a,u),(v,c)\in E(T)$ and $(b,u),(v,d)\notin E(T)$, i.e., $E(T)$ contains exactly two edges adjacent to $(u,v)$. Then, again after relabeling there exists a directed path $P(b,u)\in\mathcal{P}$. Analogously, we can relabel $d\neq v$ such that there exists a directed path $P(v,d)$ in $N$ with only its start- and endpoint in $V(T)$. To ensure that the endpoint of $P(v,d)$ is in $V(T)$ we use the fact that $T$ and $N$ have the same leaf set. This means, $\eta^*(N)$ can increase after applying $\sigma$ on $(u,v)$ only if the subtree $T_{\sigma}$ of $\sigma(N)$ we obtain from $T$ is a maximum covering subtree with $(u,v),(v,u)\notin E(T_{\sigma})$. Then, $\eta^*(\sigma(N))=\eta^*(N)+1$. The same argument holds for $\sigma'$ instead of $\sigma$. However, if we exchange either the roles of $a$ and $b$ or the roles of $c$ and $d$, then $\eta^*(\sigma'(N))=\eta^*(N)$ or $\eta^*(\sigma(N))=\eta^*(N)$, respectively. Thus, our claim follows.
    
\end{proof}

\cref{lem::switching} highlights a central issue when studying the CTBP for $N\in\mathcal{N}$: covering subtrees are a genuine global substructure of the network while rNNI operations are local. \cref{lem::switching}.1 was derived from the maximality of the covering subtree $T$. \cref{exmp::eta} is an application of this characterization. However, the upper bound on the CTBP of $\sigma(N)$ in \cref{lem::switching}.2 can be arbitrarily large because we cannot derive a sufficient condition for the maximality of covering subtrees in $\sigma(N)$ from rNNI operation $\sigma$ alone. This means, we need structural constraints on the neighborhood of edges affected by a rNNI operation to quantify the relationship between the CTBP and NTBP further. 
The following result is derived from such a constrained approach:

\begin{figure}[!t]
    \centering
    \includegraphics[width=0.75\linewidth]{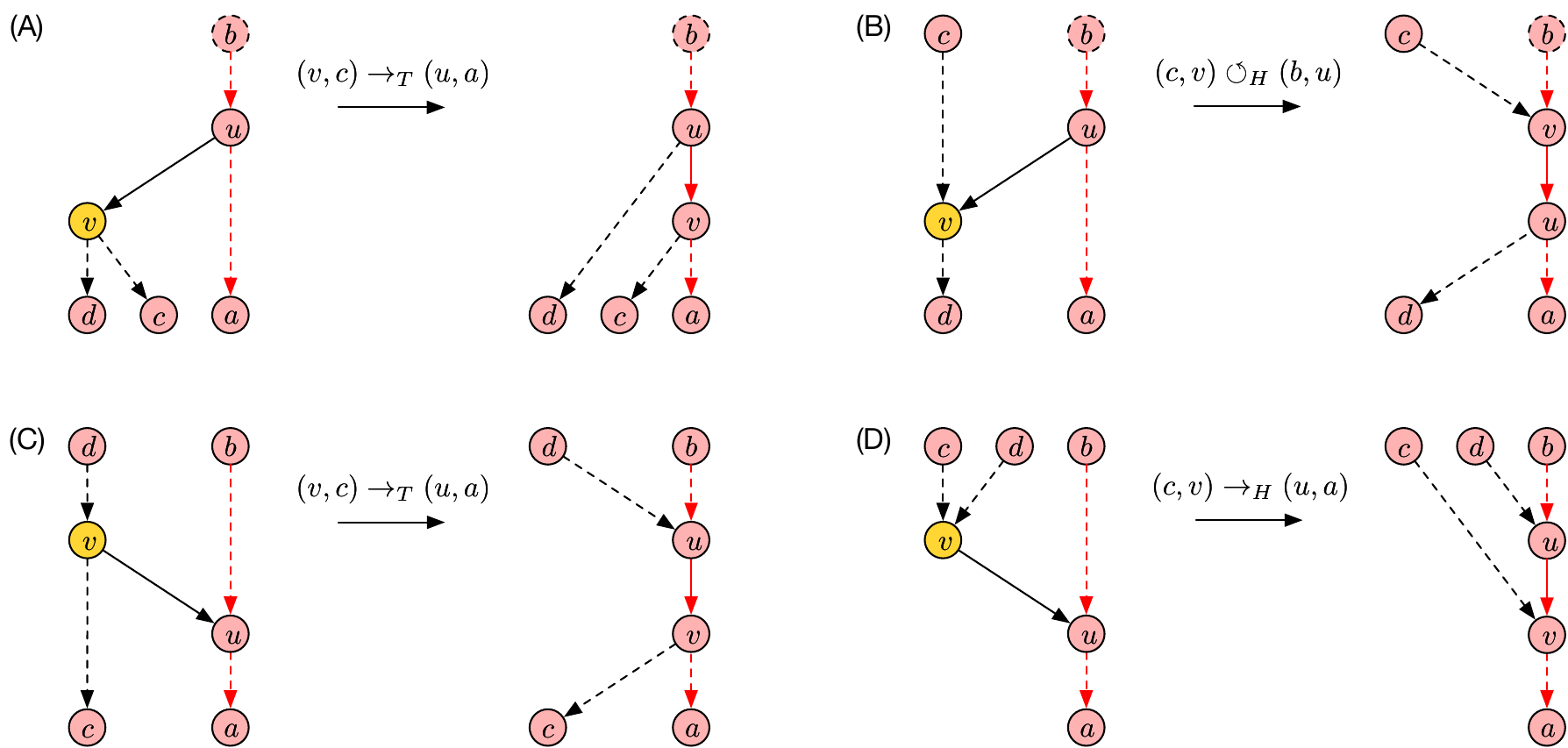}
    \caption{Four rNNI operations for the respective graphs shown. Red vertices and edges indicate a maximum covering subtree in the proof of \cref{prop:rNNI:2}.}
    \label{fig:etaUB}
\end{figure}

\begin{proposition}\label{prop:rNNI:2}
    Let $N\in\mathcal{N}$. If $N$ is not tree-based, then there exists a rNNI operation $\sigma$ such that $\eta^*(\sigma(N))=\eta^*(N)-1$.
\end{proposition}
\begin{proof}
    Assume $N$ is not tree-based. 
    and let $T^*$ be shown in red in \cref{fig:etaUB}. First, assume $(u,v)\in E(N)$ and $d_{\text{out}}(v)=2$. Then, the edges incident to $(u,v)$ are shown in the graph on the left in \cref{fig:etaUB}(A). Let $\sigma=(v,c)\to_T(u,a)$. Then, $\sigma(N)$ is depicted in the graph on the right in \cref{fig:etaUB}(A). Observe that $\sigma$ replaces edge $(u,a)$ by edges $(u,v)$ and $(v,a)$ to transform $T^*$ into a covering subtree $T_{\sigma}$ of $\sigma(N)$ with $\eta(\sigma(N),T_{\sigma})=\eta^*(N)-1$. Since the paths $P(v,c)$ and $P(v,d)$ in $N$ are not in $T^*$, the maximality of $T^*$ implies that $T_{\sigma}$ is a maximum covering subtree of $\sigma(N)$. Thus, $\eta^*(\sigma(N))=\eta^*(N)-1$.
    
    Next, assume $(u,v)\in E(N)$, $d_{\text{out}}(v)=1$. Then, the edges incident to $(u,v)$ are shown in the graph on the left in \cref{fig:etaUB}(B). Let $\sigma=(c,v)\circlearrowleft_H(b,u)$. Since $P(c,d)$ is not a subpath of $T^*$, we can apply \cref{lem::switching}.2 for reversed roles of $N$ and $\sigma(N)$ to obtain $\eta^*(N)=\eta^*(\sigma(N))+1$. The remaining cases shown in \cref{fig:etaUB}(C) and~(D) are analogous to \cref{fig:etaUB}(B) and~(A), respectively.
\end{proof}

Observe that the relationship between the NTBP and CTBP described by \cref{prop:rNNI:2} is significantly simpler than the necessary and sufficient conditions in \cref{prop:rNNI:1} to characterize the relationship between the NTBP and ATBP. Moreover, while there exists no rNNI operation $\sigma$ such that $\delta^*(\sigma(N))\leq\delta^*(N)-2$, \cref{lem::switching} tells us that $\eta^*(\sigma(N))\leq\eta^*(N)-k$ is possible for any integer $k\geq 2$. This qualitative difference between the ATBP and CTBP relative to the NTBP stems in part from the fact that a maximum covering subtree of a rooted binary phylogenetic network is not unique while the set of maximal W-fences is unique (see \cref{thm::TB}). 

Now, in light of \cref{thm:saturated}, an immediate consequence of \cref{prop:rNNI:1} and \cref{prop:rNNI:2} is the following:
\begin{corollary}\label{cor:bounds}
    Let $N\in\mathcal{N}$. Then, $\delta^*(N)\leq\delta_{\text{rNNI}}(N)\leq\eta^*(N)$. Both inequalities hold with equality if and only if there exists a $\mathcal{W}_N$-saturated M-W-matching in $N$.
\end{corollary}

Combining \cref{cor::proximitySummary} and \cref{cor:bounds}, we now know that \cref{prop::rNNISummary} in the beginning of this section holds, too. Moreover, recall from the previous subsection that \cref{thm:saturated} and \cref{cor:rNNI:1} imply that there exist $N\in\mathcal{N}$ with $\delta^*(N)=\delta_{\text{rNNI}}(N)<\eta^*(N)$. Indeed, we quantify this special case of \cref{cor:bounds} as follows:

\begin{figure}[!b]
    \centering
    \includegraphics[width=0.5\linewidth]{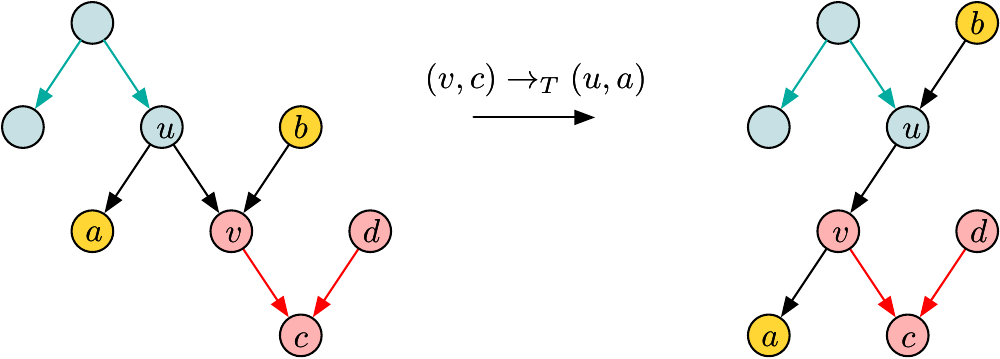}
    \caption{The graph on the right can be obtained from the graph on the left by the shown rNNI operation.
    }
    \label{fig:Nfenceproof}
\end{figure}

\begin{figure}[!t]
    \centering
    \includegraphics[width=0.8\linewidth]{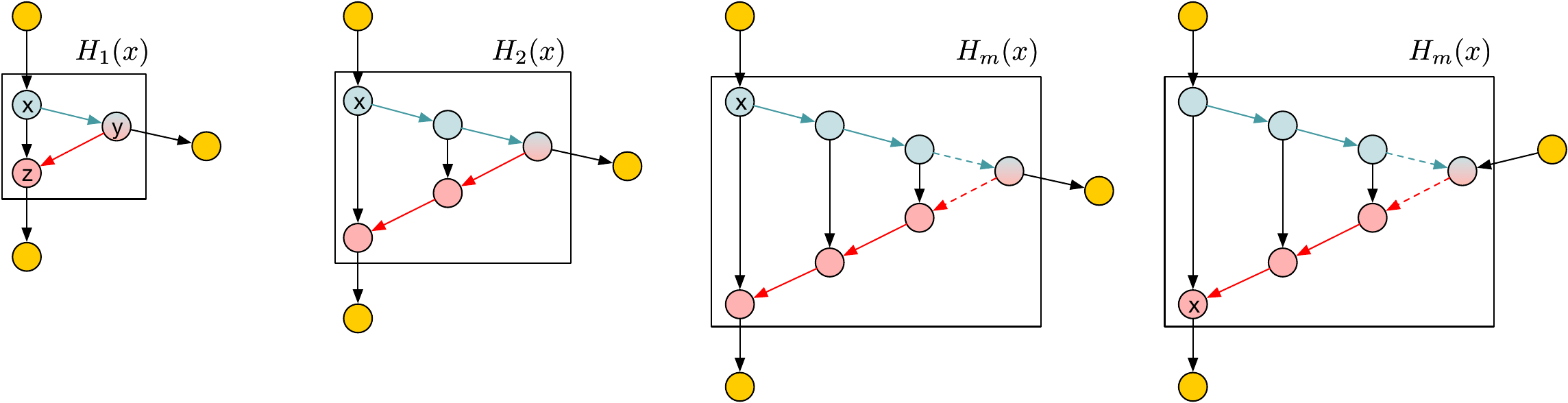}
    \caption{From left to right: $H_1(x)$ depicts the directed subgraph obtained from an internal tree vertex $x$ of a rooted binary phylogenetic network by inserting a new edge $(y,z)$ such that $y$ and $z$ are the children of $x$. $H_2(x)$ is constructed from $H_1(x)$ by inserting another edge for internal tree vertex $y$. For any positive integer $m$, inductive repetition of this edge insertion yields $H_m(x)$. If $x$ is a reticulation vertex then $H_m(x)$ is constructed analogously and shown on the right.
    }
    \label{fig:Nfenceproof2}
\end{figure}

\begin{restatable}{proposition}{propetak}\label{prop::eta-k}
    For any integer $k>0$ there exists $N\in\mathcal{N}$ such that $\delta^*(N)=\delta_{\text{rNNI}}(N)=\eta^*(N)-k.$
\end{restatable}
\begin{proof} 
    First, let $k$ be a positive integer and let $N\in\mathcal{N}$ such that there exists a $\mathcal{W}_N$-saturated M-W-rNNI-matching in $N$ for which \cref{prop:rNNI:1}.2 is satisfied at least once. This means, there exists a distance-1 M-W-pair $(Y,Z)$ such that its connecting edge $(u,v)$ is incident to the endpoint of $Z$. Choose $N$ such that $u$ is the endpoint of $Y$ and $(u,v)$ is not in a W-fence in $\mathcal{Z}(N)$. Then, there exist $a,b\in V(N)$, $a\neq v$, $b\neq u$, such that $(u,a),(u,v),(b,v)\in E(N)$. The graph on the left in \cref{fig:Nfenceproof} depicts $Y$ and $Z$ in blue and red, respectively. Then, we consider the rNNI operation $\sigma=(v,c)\to_T(u,a)$ in the same figure. Now, we show how $N$ can be modified to achieve $\eta^*(\sigma(N))=\eta^*(N)-k$. Afterwards we will modify $N$ further to maintain the existence of a $\mathcal{W}_N$-saturated M-W-rNNI-matching in $N$.
    
    \begin{description}
        \item[Case 1:] There exists a maximum covering subtree $T^*$ of $N$ such that $(u,a)\in E(T^*)$ and $(u,v),(b,v)\notin E(T^*)$. Let $v'\in V(T^*)$, $v'\neq u$, such that there exists a directed path $P(v',v)$ with $E(P(v',v))\subset E(N)\setminus E(T^*)$. Necessarily $b\in P(v',v)$ because $u\in V(T^*)$. Then, there exists a covering subtree $T_{\sigma}$ of $\sigma(N)$ we obtain from $T^*$ by 
        \begin{enumerate}
            \item removing the maximal length subpath $P$ of $T^*$ with endpoint $u$ consisting only of edge subdivisions of $T^*$;
            \item adding $P(v',v)$ by subdividing its last edge by $u$ (see \cref{fig:Nfenceproof} for last edge $(b,v)$). 
        \end{enumerate}
        Let $p=|V(T^*)\setminus V(T_{\sigma})|$ and $q=|V(T_{\sigma})\setminus V(T^*)|=|V(P(v',v))|-2$. Then, we have
        \begin{align*}
            \eta(\sigma(N),T_{\sigma})=\eta(N,T^*)+p-q.
        \end{align*}
        If $k=q-p$, then $\eta^*(\sigma(N))=\eta(\sigma(N),T_{\sigma})=\eta^*(N)-k$ follows from the maximality of $T^*$. We show how~$N$ can be modified to draw the same conclusion for $k\neq q-p$. Without loss of generality $k>q-p$. Otherwise we reverse the role of paths $P(v',v)$ and $P$ in the following. Then, for a positive integer~$m$, we replace an internal vertex of $P(v',v)$, say vertex~$b$, in $N$ by the auxiliary graph $H_m(b)$ shown in~\cref{fig:Nfenceproof2}. Subsequently, $T^*$ might no longer be maximal. In this case the red and/or blue paths in~\cref{fig:Nfenceproof2} are subpaths of a maximum covering subtree $T'$ of $N$ with maximal overlap $|V(T')\cap V(T^*)|$. 
        
        If the blue path is in $T'$ but not the red path, then we replace $T^*$ by $T'$ and all previous arguments for $T^*$ still apply. Indeed, again $(b,v),(v,c)\notin E(T^*)$, $p$ cannot increase and $q$ is a function $q(m)$ increasing in $m$. This means, we can prove our claim by choosing $q(m)=k+p$ for a sufficiently large $m$.

        If the red path is in $T'$, then $T'$ also contains the blue path by the maximality of $T'$. Now, $(b,v),(v,c)\in E(T')$ because $T'$ and $N$ have the same leafset. Then, there exists a directed subpath $P(v,c')$ of $T'$ of minimal length such that $c'\in V(T')\cap V(T^*)$. For a positive integer $m'$, we replace parent $w\in V(T^*)\setminus V(T')$ of $c'$ in $N$ by auxiliary graph $H_{m'}$. By choosing $m'$ sufficiently large we can ensure that $T'$ is no longer a maximum covering subtree of $N$. Hence, no maximum covering subtree of $N$ contains the red path of $H_m$ and we can return back to our previous case to consider $T^*$ as a maximal covering subtree of $N$ which contains the blue path but not the red path.

        \item[Case 2:] There exists a maximum covering subtree $T^*$ of $N$ such that $(u,v)\in E(T^*)$ and $(b,v)\notin E(T^*)$. Analogous to Case~1 because $(b,v)\notin E(T^*)$ and $(v,c)\in E(T^*)$.
    
        \item[Case 3:] There exists a maximum covering subtree $T^*$ of $N$ such that $(b,v)\in E(T^*)$. Analogous to Case~1 or~2 if $(u,a)\notin E(T^*)$ by reversing the direction on all arguments. Hence, assume $(u,a)\in E(T^*)$. Then, $(d,c)\in E(N)\setminus E(T^*)$ for a vertex $d\neq v$. Subdivide edges $(b,v)$ and $(d,c)$ by vertices $x$ and $y$, and add edge $(x,y)$ to $N$. Then, path $P(b,c)$ in $T^*$ can be replaced by the path constituted by edges $(b,x)$, $(x,y)$, $(y,c)$ to obtain a maximum covering subtree $T'$ of $N$. Since $(x,v)$ takes the role of $(b,v)$ with respect to rNNI operation $\sigma$ and $(x,v),(v,c)\notin E(T')$, Case~1 applies.
        %
        
        \item[Case 4:] There exists a maximum covering subtree $T^*$ of $N$ such that $(u,a),(u,v),(b,v)\notin E(T^*)$. Then, for a positive integer~$m$, replace $b$ in $N$ by the auxiliary graph $H_m(b)$. For sufficiently large $m$, Case~3 applies.
    \end{description}

\begin{figure}[!t]
    \centering
    \includegraphics[width=0.8\linewidth]{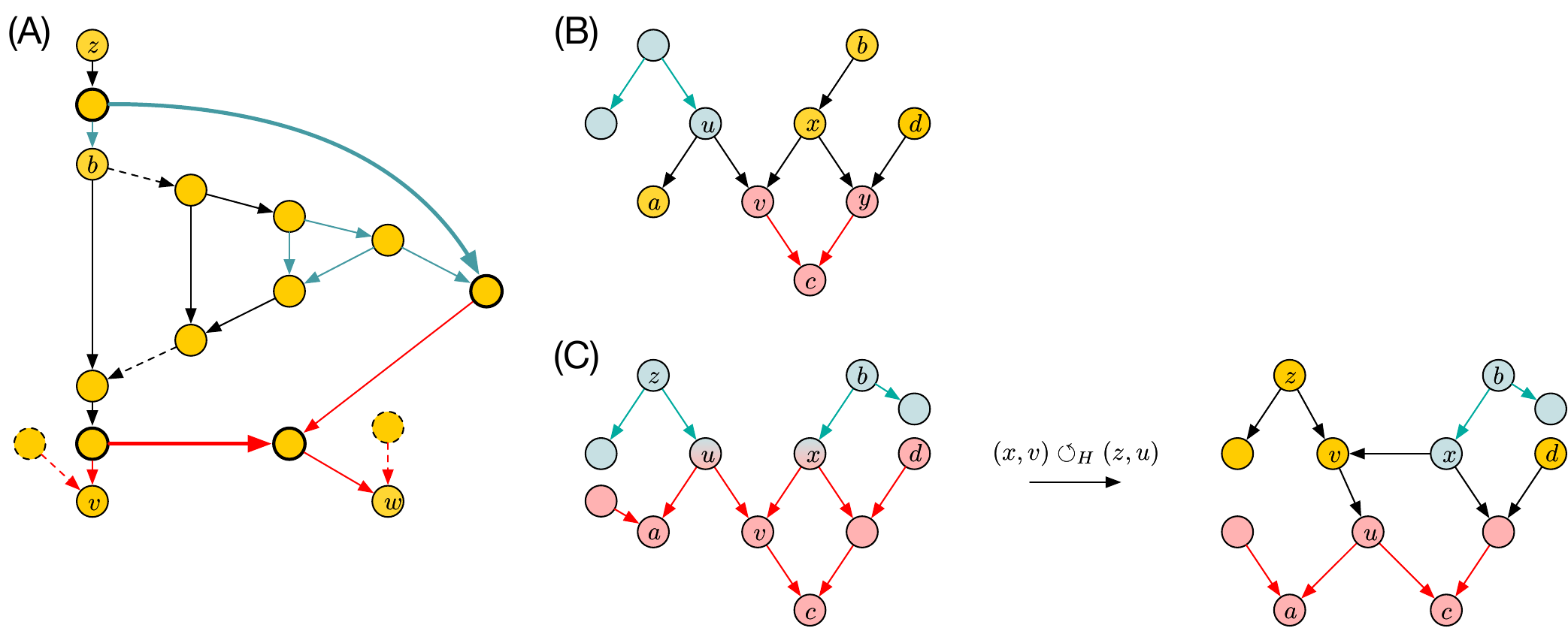}
    \caption{(A) A modification of auxiliary graph $H_m(b)$; (B) a modification of the graph on the left in \cref{fig:Nfenceproof}; (C) a rNNI operation on a modification of the graph on the left in \cref{fig:Nfenceproof}. Blue and red vertices and edges refer to M-fences and W-fences, respectively, in the proof of \cref{prop::eta-k}.}
    \label{fig:Nfenceproof3}
\end{figure}
    
    Finally, observe that there does not exist a $\mathcal{W}_{\sigma(N)}$-saturated M-W-rNNI-matching in $\sigma(N)$ only if the insertion of auxiliary graph $H_m(b)$ or edge $(x,y)$ in Case~3 does change any of conditions in \cref{prop:rNNI:1}.1 to~\cref{prop:rNNI:1}.5 satisfied by the original network. Indeed, without the insertion of graph $H_m(b)$ or edge $(x,y)$ notice that our original M-fence $Y$ and W-fence $Z$ either are both contained in a M-fence and W-fence in $\mathcal{Z}(\sigma(N))$, respectively, or are both a N-fence in $\mathcal{Z}(\sigma(N))$. 

    First, consider the insertion of $H_m(b)$. If $b$ is a reticulation vertex in $N$, then the W-fence in $\mathcal{Z}(N)$ containing $b$ as an internal vertex is decomposed into two N-fences through the insertion of $H_m(b)$ (see \cref{fig:Nfenceproof2}). Otherwise $b$ is a tree vertex. Then, we insert the edges and vertices shown in bold in \cref{fig:Nfenceproof3}(A) into $H_m(b)$. This means, $b$ is the endpoint of a M-fence $A$ (shown in blue) and either $v$ or $w$ is an internal vertex of a W-fence $B$ (shown in red). Indeed, if $v$ and $w$ are in a W-fence in $\mathcal{Z}(N)$, then $(u,v)$ is in a W-fence, contradicting our initial choice of $N$. This means, $(A,B)$ is either a M-W-pair or a distance-1 M-W-pair such that its connecting edge is incident to the endpoint of $B$, ensuring the existence of a $\mathcal{W}_N$-saturated M-W-rNNI-matching.

    Next, consider the insertion of edge $(x,y)$ in Case~3. Since initially edges $(u,a)$, $(u,v)$ and $(b,v)$ are in the same N-fence or crown, the graph on the left in either  \cref{fig:Nfenceproof3}(B) or (C) correctly depicts M-fences and W-fences in blue and red, respectively. In case (B) edges adjacent to $(u,v)$ have the same relationships to M-fences and W-fences as for the initial choice of $N$ (see \cref{fig:Nfenceproof}). In case (C) the existence of a $\mathcal{W}_N$-saturated M-W-rNNI-matching is guaranteed, too. To see this consider the effect of the rNNI operation shown in \cref{fig:Nfenceproof3}(C) on W-fences.
\end{proof}

To conclude this section we consider an example to illustrate the relation $\delta^*(N')<\delta_{\text{rNNI}}(N')<\eta^*(N')$ derived from \cref{thm:saturated} and \cref{prop::delta+k}.

\begin{figure}[!b]
    \centering
    \includegraphics[width=0.6\linewidth]{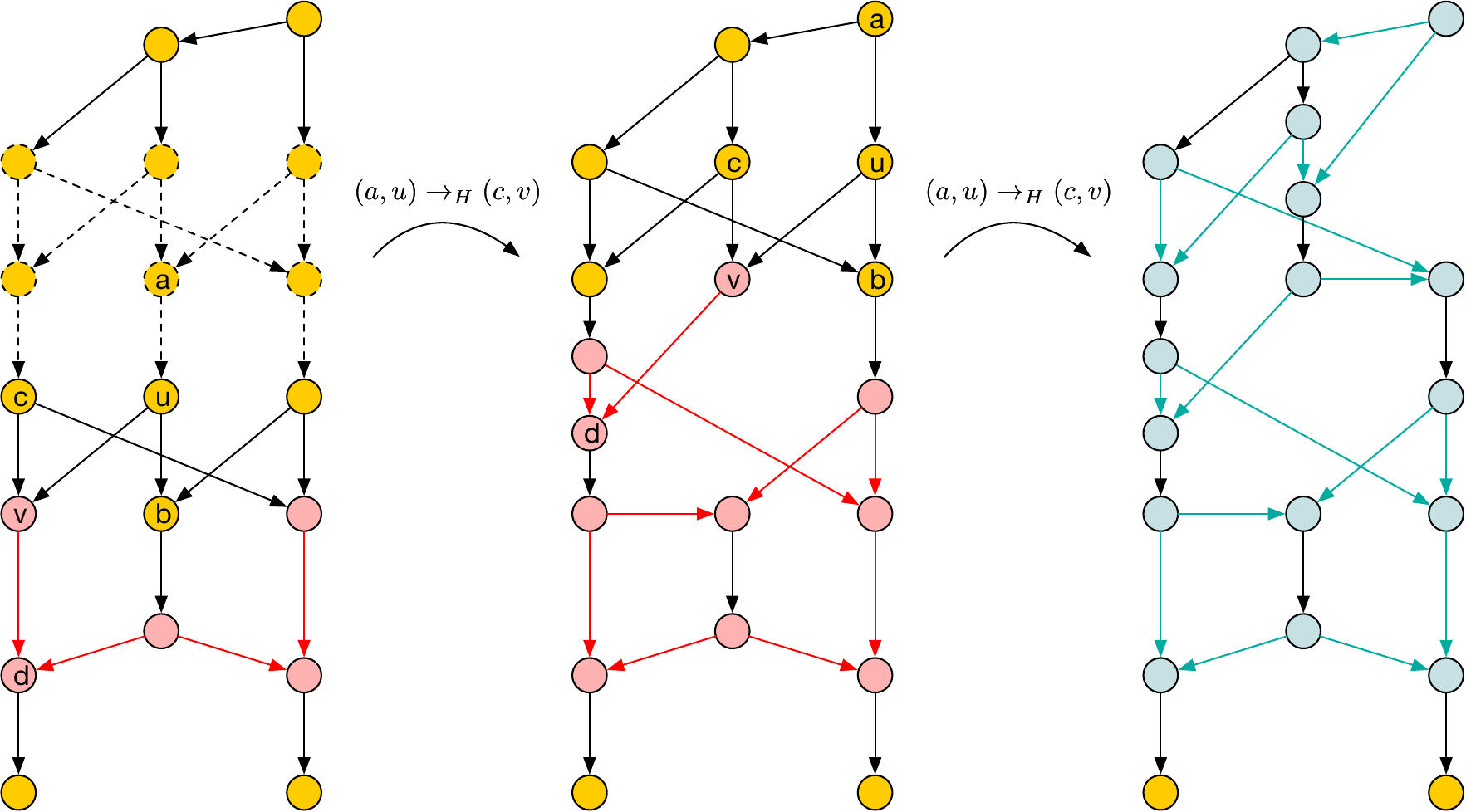}
    \caption{Three rooted binary phylogenetic networks which differ by the rNNI operations shown. All W-fences in the decomposition of each network into maximal zig-zag trails are shown in red. A N-fence is shown in blue.}
    \label{fig:Wfence-failure}
\end{figure}

\begin{example}\label{exmp::delta}
Consider the networks $N_1,N_2,N_3$ from left to right in Figure~\ref{fig:Wfence-failure}. Note that $N_3$ is tree-based but $N_1$ and $N_2$ are not. Like in \cref{fig:covering-failure} the crown and N-fences indicated by dashed lines in $N_1$ can be duplicated multiple times to increase the path from the root to the leaves to any arbitrary but fixed length. Observe that $N_2$ contains a $\mathcal{W}_N$-saturated M-W-rNNI-matching because edge $(u,v)$ fits into \cref{prop:rNNI:1}.2. Moreover, $N_1$ contains no $\mathcal{W}_N$-saturated M-W-rNNI-matching because no pair of M-fences and W-fences fits into Proposition~\ref{prop:rNNI:1}. Hence,
\begin{align*}
    \delta^*(N_1)=1<\delta_{\text{rNNI}}(N_1)=2<\eta^*(N_1)=4
\end{align*}
where the last equation follows from the same observation as in \cref{exmp::eta} due to the similarity of the network constructions in \cref{fig:covering-failure} and~\cref{fig:Wfence-failure}. Thus, by extending $N_1$ with crowns and N-fences $\delta_{\text{rNNI}}(N_1)$ can be made arbitrarily large while $\delta^\ast(N_1)=1$ and $\eta^*(N_1)=2\cdot\delta_{\text{rNNI}}(N_1)$ holds.
\end{example}

\subsection{A neighborhood-based exact solution algorithm for the NTBP}\label{sec:algo}
\begin{algorithm}[!t]
\caption{An exact solution algorithm for the NTBP}\label{alg:bruteforce-rntbp}
\Input{$N\in\mathcal{N}$}
\Output{$\delta_{\text{rNNI}}(N)$}
Compute $\eta^*(N)$\;
$\mathcal{B}_0(N)\leftarrow\{N\}$\;
\For{$k=1,\dots,\eta^*(N)$}{
    $\mathcal{B}_k(N)\leftarrow\bigcup\left\{\mathcal{B}_1(N')\,:\,N'\in\mathcal{B}_{k-1}(N)\right\}$\;
    \If{$\mathcal{B}_k(N)\cap\mathcal{N}_{TB}\neq\emptyset$}{
        \Return{$k$}
    }
}
\end{algorithm}
Lastly, we provide an exact solution method for the NTBP. To this end, for $N\in\mathcal{N}$ and an integer $k$, we denote $\mathcal{B}_k(N)$ as the set of $k$-neighborhoods of $N$ with respect to rNNI operations, i.e., the set of rooted phylogenetic networks that can be obtained by applying at most $k$ rNNI operations to $N$. For $N\in\mathcal{N}$, let
\begin{align*}
    m_{ij}&=\left|\left\{(u,v)\in  E(N)\,:\,d_{\text{in}}(u)=i,\,d_{\text{in}}(v)=j\right\}\right|,
\end{align*}
let $m=|E(N)|$, $n=|V(N)|$ and let $r$ denote the number of reticulation vertices in $N$. Then, the 1-neighborhood of $N$ is bounded as follows:

\begin{proposition}[Proposition~2 in S1 text in~\cite{gambette2017rearrangement}]\label{prop::B}
    Let $N\in\mathcal{N}$. Then, $$|\mathcal{B}_1(N)|\leq 2(m_{11}+m_{22})+3m_{12}+4m_{21}.$$
\end{proposition}

\begin{proposition}\label{prop::algo}
    \cref{alg:bruteforce-rntbp} solves the NTBP exactly in $\mathcal{O}\left(n^2m+(3m+r)^{\eta^*(N)}\right)$ time.
\end{proposition}
\begin{proof}
    We know from \cref{cor:bounds} that $\delta_{\text{rNNI}}(N)\leq\eta^*(N)$ and $\eta^*(N)$ can be computed in $\mathcal{O}(n^2m)$ time~\cite{davidov2020maximum}. Hence, \cref{alg:bruteforce-rntbp} terminates. Furthermore, from \cref{prop::B} we know that
    \begin{align*}
        |\mathcal{B}_1(N)|\leq 2(m_{11}+m_{22})+3m_{12}+4m_{21}<3(m_{11}+m_{12})+4(m_{22}+m_{21})=3(m-r)+4r=3m+r
    \end{align*}
    and therefore $|\mathcal{B}_k(N)|\leq|\mathcal{B}_1(N)|^k\leq(3m+r)^k.$
\end{proof}

Notice that the construction of $\mathcal{B}_k(N)$ in \cref{alg:bruteforce-rntbp} allows us to backtrack in line~6 to also identify a sequence of $k$ networks $\hat{N}_1,\dots,\hat{N}_k\in\mathcal{N}$ such that $\hat{N}_1=N$, $\hat{N}_k\in\mathcal{N}_{\text{TB}}$ and $\hat{N}_i\in\mathcal{B}_1(\hat{N}_{i-1})$ for all $i\in\{2,\dots,k\}$. However, since $\eta^*(N)\leq 2r-2$ (Lemma 5.20 in~\cite{janssen2021rearranging}), the runtime guarantee of \cref{prop::algo} is exponential in the number of reticulation vertices. This means, only if $r$ grows logarithmic in the input length, i.e., for a family of input networks which is close to tree-based, then \cref{alg:bruteforce-rntbp} solves the NTBP exactly in polynomial time. In the special case that there exists a $\mathcal{W}_N$-saturated M-W-rNNI-matching in~$N$ (\cref{cor:rNNI:1}) we can strengthen this result because then the NTBP of~$N$ can be solved exactly in linear time~\cite{hayamizu2021structure}.

\section{Conclusion and future research}\label{sec:future}
In this article we have conducted a comparative study of the $\delta_{\text{rNNI}}$ to derive \cref{prop::rNNISummary}, a linear order of all known proximity measures for tree-basedness of rooted binary phylogenetic networks. The same result might hold if we remove the restriction to binary networks because our proof rests on \cref{thm::TB} which has been extended to the non-binary case~\cite{pons2025fence}. However, this extended decomposition of a rooted phylogenetic network into maximal zig-zag trails is characterised by three forbidden subgraphs and the differentiation of crowns into five pairwise distinct types of maximal zig-zag paths (see Figure 3 in~\cite{pons2025fence}). Therefore, a generalization of \cref{prop:rNNI:1} to rooted non-binary phylogenetic networks might be achievable when working with the extended maximal zig-zag trail decomposition. Though we think a substantially longer proof and possibly proof automation are necessary. We also notice that the application of \cref{prop::rNNISummary} leads to a combinatorial explosion of candidates for an optimal solution to $\delta_{\text{rNNI}}$ (see \cref{cor:rNNI:1} and \cref{exmp:rNNI:1}). This complicates a potential characterization of a family of networks in $\mathcal{N}$ for which $\delta^*(N)=\delta_{\text{rNNI}}(N)$ holds. The same challenge persists in the construction of networks in $\mathcal{N}$ for which rNNI moves tighten neither lower bounds $\delta^*(N)$ nor upper bounds $\eta^*(N)$ particularly well. Specifically, for any positive integers $k_1,k_2$, we have shown that there exists a family $\mathcal{F}$ of networks $N\in\mathcal{N}$ such that $\delta^*(N)+k_1=\delta_{\text{rNNI}}(N)=\eta^*(N)-k_2$ by first applying \cref{prop::delta+k} and then \cref{prop::eta-k}. Here the lack of uniqueness of maximum covering subtrees affected by a rNNI operation (see \cref{lem::switching}) in contrast to the effect of rNNI operations on W-fences (see \cref{prop:rNNI:1}) differentiates the claim of \cref{prop::delta+k} for all networks in $\mathcal{N}$ from the network existence claim in \cref{prop::eta-k}. Similar to the question on generalizing \cref{prop:rNNI:1} from binary to non-binary networks, the construction of $N\in\mathcal{N}$ in the proof of \cref{prop::eta-k} might still be possible for any given $\mathcal{W}_N$-saturated M-W-rNNI-matching with a substantially longer proof. While the complexity of the NTBP is an open question, we know from \cref{prop::algo} that for a sufficiently small $k_2$ we can compute $\delta_{\text{rNNI}}$ in polynomial time because $\eta^*(N)$ can be computed efficiently. Hence, for sufficiently large $k_1$ and $k_2$, family $\mathcal{F}$ seems to contain the input networks for which the NTBP might be hard to solve. 

Further future research directions concern the boundaries of our study. For example, one can consider the NNI proximity measure of tree-basedness for unrooted instead of rooted phylogenetic networks, whose decision version is NP-complete~\cite{fischer2020tree} to further study its relation with comparable proximity measures. Furthermore, one can consider more general topological rearrangement operations than rNNI such as Subnet Prune and Regraft (SNPR) operations~\cite{bordewich2017lost} to study proximity from tree-basedness in rooted phylogenetic networks. This generalization is not trivial as SNPR operations can reconnect vertices at large topological distances, allowing for large distances in M-W-pairs or W-M-pairs to describe or characterize the ATBP through SNPR moves while ensuring acyclicity of the network. This leads to more relevant subgraphs in an analysis of SNPR moves similar to \cref{prop:rNNI:1}. Lastly, one can define a rNNI proximity measure from orchard networks, a subclass of tree-based networks, or tree-child networks, a subclass of orchard networks~\cite{kong2022classes}. To the best of our knowledge both measures have not been studied before and might be interesting because the analogous version of the related ATBP for orchard or tree-child network is NP-hard and linear-time solvable, respectively~\cite{van2023making}. This question also includes the study of new upper bounds because the CTBP is derived from the definition of tree-basedness. For example, one could study the minimum number of vertices in a network who do not have a tree vertex as their child as a proximity measure for tree-childness to maybe obtain upper bounds on the rNNI tree-childness proximity measure.

\paragraph{Funding.} This work was supported by grant OCENW.M.21.306 of the Dutch Research Council (NWO).

\bibliographystyle{plainnat}

\appendix

\section{Omitted proofs}\label{sec:appendix}

\propdeltaSupermodular*
\begin{proof}
    Let $G_N$ be the bipartite graph whose vertex partition sets are copies of $V(N)$ and an edge $\{u,v\}\in E(G_N)$ joins the first copy of $u$ with the second copy of $v$ if and only if $(u,v)\in E(N)$. Let $M$ be a maximum cardinality matching in $G_N$. Then, we know from Lemma 4.1 in~\cite{francis2018new} that
    \begin{align*}
        \delta^*(N)=|\left\{u\in V(G_N)\,:\,\exists(u,v)\in E(G_N)\text{ with }\{u,v\}\notin M\right\}|-|X|
    \end{align*}
    Recall that, for $W\subseteq V(N)$, $\Delta W=\{v\in V(N)\,:\,\exists\,(w,v)\in E(N)\text{ with }w\in W\}$. Since $\Delta W$ is the set of neighbors of $W$ in $G_N$ if we associate $W$ with the first vertex partition set of $G_N$, we know from Hall's marriage theorem~\cite{plummer1986matching} that
    \begin{align*}
        |M|=|V(N)|-\max_{W\subseteq V(N)}(|W|-|\Delta W|).
    \end{align*}
    Since $|\left\{u\in V(G_N)\,:\,\exists(u,v)\in E(N)\text{ with }\{u,v\}\notin M\right\}|=|V(N)|-|M|$, we obtain
    \begin{align}\label{eq::deltaWX}
        \delta^*(N)=\max_{W\subseteq V(N)}(|W|-|\Delta W|)-|X|.
    \end{align}
    Observe that $(u,v)\notin E(G_N)$ for all $u\in X$. Hence, $\Delta X=\emptyset$, i.e., every maximizer of~\eqref{eq::deltaWX} contains $X$. Thus, our claim follows.
\end{proof}

\propkappaRatio*
\begin{proof}
Recall that, for $W\subseteq V(N)$, $\Delta W=\{v\in V(N)\,:\,\exists\,(w,v)\in E(N)\text{ with }w\in W\}$. First, we show that
\begin{align*}
    \kappa^*(N) =\max_{W\subseteq V\setminus X}\left\lceil\frac{|W\setminus\Delta W|}{|\Delta W\setminus W|}\right\rceil.
\end{align*}

    Let $T$ be a covering subtree of $N$, let $W\subseteq V\setminus X$ and let $v\in V(T)\cap W\setminus \Delta W$. Since $V(T)$ contains every leaf but $W$ contains none, for every directed path $P$ in $T$ from $v$ to a leaf of $N$, there exists a first edge $(a,b)\in E(P)$ such that $a\in W$ and $b\notin W$. By definition, $b\in\Delta W\setminus W$. Let $\phi:V(T)\cap W\setminus\Delta W\to V(T)\cap\Delta W\setminus W$ be defined by $\phi(v)=b$ if and only if there exists a directed path $P$ from $v$ to $b$ in $T$ such that $V(P)\setminus\{b\}\subseteq W$. We show that $\phi$ is injective. Suppose by contradiction there exist two pairwise distinct vertices $v,v'\in V(T)\cap W\setminus\Delta W$ such that $\phi(v)=\phi(v')=b$. Since there exist directed paths from both $v$ and $v'$ to $b$ in $T$ and $T$ is a tree, there exists without loss of generality a directed path $P'$ from $v$ to $v'$ in $T$. However, $V(P')\subseteq W$ by our choice of $b$. This leads to a contradiction because $v'\in W\setminus\Delta W$ but $V(P')\setminus\{v\}\subseteq\Delta W$ by definition of $\Delta W$. Thus, $\phi$ is injective.

    Next, by definition of $\kappa^*(N)$ there exist covering subtrees $T_1,\dots,T_{\kappa^*(N)}$ such that $\bigcup_{i=1}^{\kappa^*(N)}V(T_i)=V(N)$. Then,
    \begin{align*}
        |W\setminus\Delta W|\leq\sum_{i=1}^{\kappa^*(N)}\left|V(T_i)\cap W\setminus\Delta W\right|\stackrel{\phi\text{ injective}}{\leq}\sum_{i=1}^{\kappa^*(N)}\left|V(T_i)\cap \Delta W\setminus W\right|\leq \kappa^*(N)\cdot|\Delta W\setminus W|.
    \end{align*}
    Hence,
    \begin{align*}
    \kappa^*(N) \geq\max_{W\subseteq V\setminus X}\left\lceil\frac{|W\setminus\Delta W|}{|\Delta W\setminus W|}\right\rceil.
\end{align*}
    Conversely, for a positive integer $k$, construct an auxiliary graph $G_{\text{lb}}^k$ for $N$ as follows:
    \begin{enumerate}
        \item For all $v\in V(N)$, $G_{\text{lb}^k}$ contains auxiliary vertices $v_{\text{in}}, v_{\text{out}}$ and edge $(v_{\text{in}},v_{\text{out}})$;
        \item For all $(u,v)\in E(N)$, $G_{\text{lb}^k}$ contains edge $(u_{\text{out}},v_{\text{in}})$;
        \item For all $v_{\text{in}},v_{\text{out}}\in V(G_{\text{lb}}^k)$, $G_{\text{lb}^k}$ contains edges $(s,v_{\text{in}}),(v_{\text{out}},t)$ for auxiliary vertices $s,t$;
        \item For $E=E(G_{\text{lb}}^k)$, define lower bound function $l:E\to\mathbb{R}_{\geq 0}$ and upper bound function $u:E\to\mathbb{R}_{\geq 0}$ for any $s$-$t$-flow on $G_{\text{lb}}^k$ by $l(e)=1$ if $e=(v_{\text{in}},v_{\text{out}})$ and $l(e)=0$ otherwise, and $u(e)=k$ for all $e\in E$.
    \end{enumerate}
    Then, we know from Theorem 1 in~\cite{davidov2020maximum} that $G_{\text{lb}}^k$ has a feasible $s$-$t$-flow if and only if there exist pairwise distinct covering subtrees $T_1,\dots,T_k$ of $N$ due to upper bound function $u$ and $\bigcup_{i=1}^kV(T_i)=V(N)$ due to lower bound function $l$. Hence, $\kappa^*(N)$ is equal to the minimum choice of $k$ such that there exists a feasible $s$-$t$-flow for $G_{\text{lb}}^k$. For a fixed $k$, we can find a $s$-$t$-flow on $G_{\text{lb}}^k$ (if it exists) by solving a maximum-cost or minimum-cost flow problem for an uniform cost function on $E(G_{\text{lb}}^k)$ in $\mathcal{O}(|V(N)|^2|E(N)|)$ time~\cite{davidov2020maximum}. To find a minimum $k$ for which a feasible $s$-$t$-flow on $G_{\text{lb}}^k$ exists, we apply binary search. Thus, we can compute $\kappa^*(N)$ in $\mathcal{O}(|V(N)|^2|E(N)|\log|V(N)|)$ time.
    
    For $U\subseteq V(G_{\text{lb}}^k)$, let $\delta^-(U)=\{(u,v)\in E\,:\,v\in U\}$ and $\delta^+(U)=\{(u,v)\in E\,:\,u\in U\}$. Since there exists a feasible $s$-$t$-flow in $G_{\text{lb}}^{k}$ if and only if
    \begin{align*}
        \sum_{e\in\delta^-(U)}l(e)&\leq\sum_{e\in\delta^+(U)}u(e)=\sum_{e\in\delta^+(U)}k~&~&\forall\, U\subseteq V\left(G_{\text{lb}}^{k}\right)\\
        \Leftrightarrow~~~\frac{\sum_{e\in\delta^-(U)}l(e)}{|\delta^+(U)|}&\leq k~&~&\forall\, U\subseteq V\left(G_{\text{lb}}^{k}\right)
    \end{align*}
    we deduce
    \begin{align}\label{eq::kappaU}
    \max_{U\subseteq V\left(G_{\text{lb}}^{\kappa^*(N)}\right)}\left\lceil\frac{\sum_{e\in\delta^-(U)}l(e)}{|\delta^+(U)|}\right\rceil=\kappa^*(N).
    \end{align}
    Let $U^*$ be a maximizer of \eqref{eq::kappaU}. Let $W^*=\{v\in V(N)\setminus X\,:\,v_{\text{out}}\in U^*\}$. Then, $\Delta W^*=\{v\in V(N)\,:\,\exists\,(u,v)\in E(N)\text{ with }u_{\text{out}}\in U^*\}$. Consider $(u,v)\in E(N)$ such that $u_{\text{out}}\notin U^*$ and $v_{\text{out}}\in U^*$. Suppose by contradiction $v_{\text{in}}\in U^*$. Then, we can remove $v_{\text{in}}$ from $U^*$ to add the contribution $l((v_{\text{in}},v_{\text{out}}))=1$ to the nominator in \eqref{eq::kappaU} without increasing the denominator, contradicting the maximality of $U^*$. Hence, $v_{\text{in}}\notin U^*$. In other words,
    \begin{align*}
        \sum_{e\in\delta^-(U^*)}l(e)=|\delta^-(U^*)|=\left|\{v\in V(N)\setminus X\,:\,u_{\text{out}}\notin U^*,v_{\text{out}}\in U^*~\forall\,(u,v)\in E(N)\}\right|=|W^*\setminus\Delta W^*|.
    \end{align*}
    For edges $(u,v)\in E(N)$ such that $u_{\text{out}}\in U^*$ and $v_{\text{out}}\notin U^*$, analogous arguments apply ensuring $u_{\text{in}}\in U^*$. Hence, $|\delta^+(U^*)|=|\Delta W^*\setminus W^*|$. Thus, from~\eqref{eq::kappaU} it follows
    \begin{align*}
        \kappa^*(N)=\left\lceil\frac{|W^*\setminus\Delta W^*|}{|\Delta W^*\setminus W^*|}\right\rceil\leq\max_{W\subseteq V\setminus X}\left\lceil\frac{|W\setminus\Delta W|}{|\Delta W\setminus W|}\right\rceil.
    \end{align*}
    Now, we can apply \cref{prop::deltaSupermodular} to obtain
    \begin{align*}
        \kappa^*(N)-1 =\max_{W\subseteq V\setminus X}\left\lceil\frac{|W\setminus\Delta W|-|\Delta W\setminus W|}{|\Delta W\setminus W|}\right\rceil=\left\lceil\frac{\max\{|W|-|\Delta W|\,:\,W\subseteq V\setminus X\}}{|\Delta W^*\setminus W^*|}\right\rceil=\left\lceil\frac{\delta^*(N)}{|\Delta W^*\setminus W^*|}\right\rceil
    \end{align*}
    for some non-empty set $W^*\subseteq V\setminus X$.
\end{proof}

\begin{figure}[!b]
    \centering
    \includegraphics[width=0.7\linewidth]{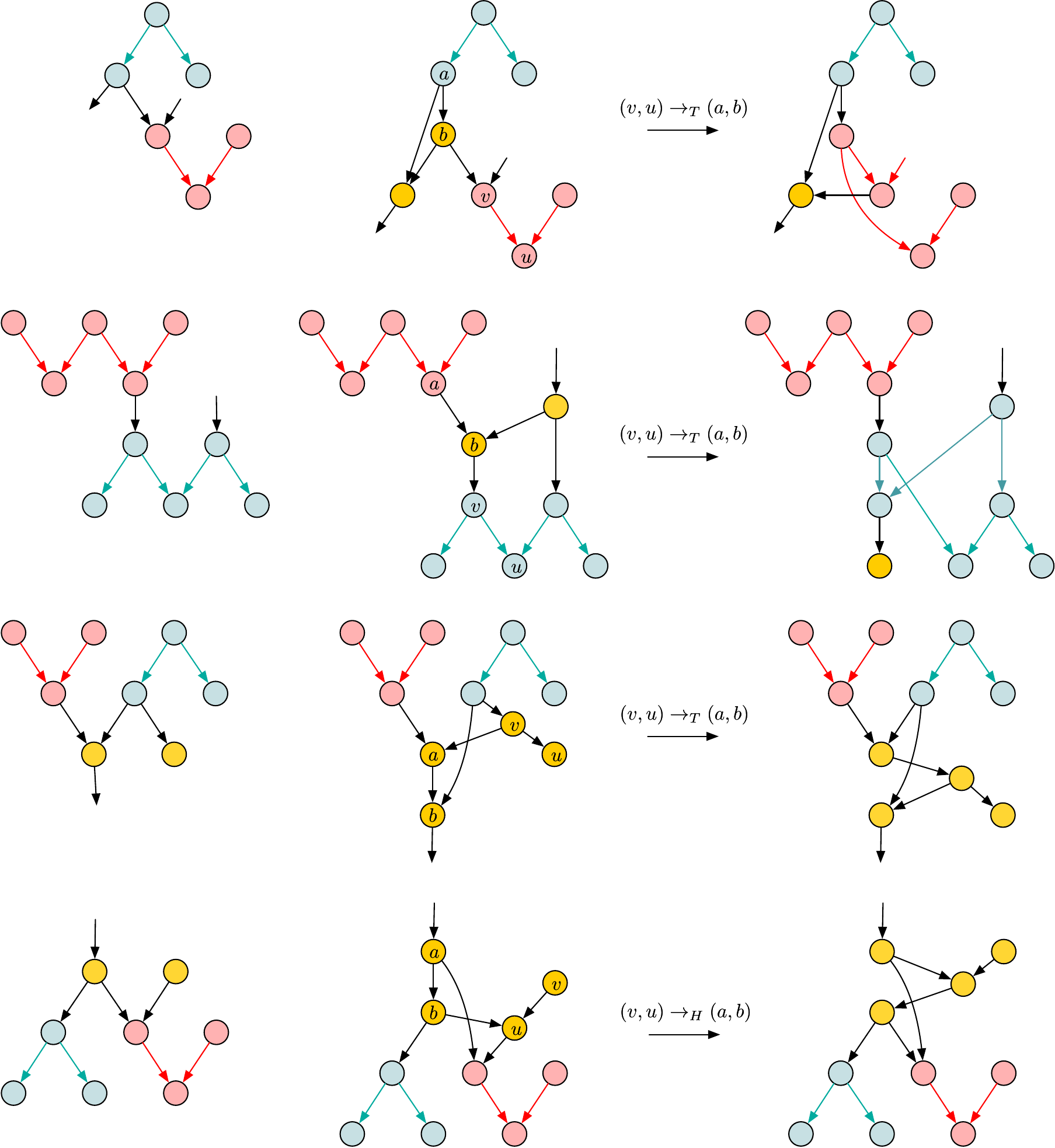}
    \caption{Graphs used in the proof of \cref{prop::delta+k}.}
    \label{fig:delta+k2}
\end{figure}


\propdeltak*
\begin{proof}
Assume there exists a $\mathcal{W}_N$-saturated M-W-rNNI-matching in $N$ for which one of \cref{prop:rNNI:1}.1 to~\cref{prop:rNNI:1}.5 is satisfied at least once. We proved our claim for the condition in \cref{prop:rNNI:1}.1. Conditions in \cref{prop:rNNI:1}.2 to~\cref{prop:rNNI:1}.5 are illustrated in the graphs on the left in \cref{fig:delta+k2} (from top to bottom in order of the conditions) by associating red vertices and edges to the W-fence $Z$ and blue vertices and edges to the M-fence $Z'$ satisfying the respective conditions. Then, the graphs in the middle depict the (possibly repeated) insertion of a few vertices and edges which increase the minimal topological distance between any pair of vertices $x\in Z$, $y\in Z'$. Observe that this construction transforms $N$ into $N'\in\mathcal{N}$ with $\mathcal{W}_{N'}=\mathcal{W}_N$ and $\mathcal{M}_{N'}=\mathcal{M}_N$ such that $Z$ and $Z'$ no longer satisfy any one of the conditions in~\cref{prop:rNNI:1}. However, applying the rNNI operation $\sigma$ shown in \cref{fig:delta+k2} to $N'$ yields the graph on the right in the same figure. In each case we can identify a W-fence and a M-fence in $\mathcal{Z}(\sigma(N'))$ (drawn in red and blue, respectively) which satisfy the same condition in \cref{prop:rNNI:1} as our initial zig-zag paths $Z$ and $Z'$ in $\mathcal{Z}(N)$. Furthermore, by construction $\sigma$ does not increase $|\mathcal{W}_{N'}|$ or $|\mathcal{M}_{N'}|$. This means, there exists a $\mathcal{W}_{\sigma(N')}$-saturated M-W-rNNI-matching in $\sigma(N')$ because there exists a  $\mathcal{W}_{N}$-saturated M-W-rNNI-matching in $N$. Thus, \cref{cor:rNNI:1} implies $\delta^*(N')+1=\delta_{\text{rNNI}}(N')$ and we can draw the same conclusion as in the case of \cref{prop:rNNI:1}.1.
\end{proof}

\end{document}